\documentclass[a4paper,12pt]{article}

\usepackage[english]{babel}
\usepackage[latin1]{inputenc}
\usepackage{etex}
\usepackage{pdfpages}
\usepackage{amsfonts}
\usepackage{amsmath}
\usepackage{amssymb}
\usepackage{amsthm}
\usepackage{appendix}
\usepackage[pagebackref,colorlinks=true]{hyperref}
\usepackage{cases}
\usepackage{enumerate}
\usepackage{caption}
\usepackage{setspace}
\usepackage{ulem}
\usepackage{enumitem}
\usepackage{scalerel}
\newcommand\Tau{\scalerel*{\tau}{T}}
\usepackage{upgreek}
\usepackage{fancyhdr}

\usepackage{manfnt} 
\usepackage{mathrsfs}
\usepackage{minitoc}
\usepackage{accents}

\usepackage{graphics}
\usepackage{graphicx}
\usepackage{hyperref}

\usepackage{lastpage}
\usepackage{latexsym}

\usepackage{color}
\usepackage{comment}

\usepackage{enumerate}

\usepackage{fancyhdr}
\usepackage[top=3cm, bottom=3cm, left=1.5cm , right=2cm]{geometry}
\newtheorem{Theorem}{Theorem}[section]
\newtheorem{Corollary}[Theorem]{Corollary}
\newtheorem{Lemma}[Theorem]{Lemma}
\newtheorem{Definition}[Theorem]{Definition}
\newtheorem{Proposition}[Theorem]{Proposition}
\newtheorem{Remark}[Theorem]{Remark}

\newtheorem{Assumption}[Theorem]{Assumption}

\newtheorem{Example}{Example}

\numberwithin{equation}{section}

\def\demi{\frac{1}{2}}

\def\cal{\mathcal}

\def\A{{\mathcal A}}

\def\C{{\mathcal C}}

\def\F{{\mathcal F}}
\def\sigR{{\cal R}}

\def\bF{{\mathbb F}}

\def\P{{\mathbb P}}

\def\R{{\mathbb R}}

\def\tU{{\tilde U}}

\def\wL^*{{\widehat L^{(\mu^*, \sigma^*)}}}
\def\wL{{\widehat L}}

\def\eps{\varepsilon}

\newcommand{\rmi}{{\rm (i) $\>\>$}}

\newcommand{\rmii}{{\rm (ii) $\hspace{1.5mm}$}}
\newcommand{\rmiii}{{\rm (iii)$\>\>$}}
\newcommand{\rmiv}{{\rm (iv)$\>\>$}}

\def\bit{\begin{itemize}}
\def\eit{\end{itemize}}

\def\beqn{\begin{eqnarray}}
\def\eeqn{\end{eqnarray}}

\def\beq*{\begin{eqnarray*}}
\def\eeq*{\end{eqnarray*}}

\def\cF{{\mathcal{F}}}

\def\F{\mathbb{F}}

\def\P{{\mathbb{P}}}

\def\R{{\mathbb{R}}}

\def\exp{{\text{exp}}}
\def\bit{\begin{itemize}}
\def\eit{\end{itemize}}

\def\bc{\begin{center}}
\def\ec{\end{center}}

\def\super { \end{document}}
\def\bcom{}
\def\edoc{\end{document}}

\usepackage[framemethod=tikz]{mdframed}
\usetikzlibrary{decorations.pathmorphing,calc}
\newcommand\wavydecor{%
    \draw[decoration={coil,aspect=0.1,segment length=5pt,amplitude=1.0pt},decorate,line width=1.5pt,black]
      (O|-P) -- (O);
}
\newmdenv[
hidealllines=true,
innerleftmargin=10pt,
innerrightmargin=0pt,
innertopmargin=0pt,
innerbottommargin=0pt,
leftmargin=-10pt,
skipabove=.5\baselineskip,
skipbelow=.5\baselineskip,
singleextra={\wavydecor},
firstextra={\wavydecor},
secondextra={\wavydecor},
middleextra={\wavydecor}
]{done}

\usepackage[linecolor=black,textwidth=25mm,textsize= footnotesize]{todonotes} 

{ \begin{list}%
	{$\bullet$}%
	{\setlength{\labelwidth}{10pt}%
	 \setlength{\leftmargin}{35pt}%
	 \setlength{\itemsep}{\parsep}{1pt}}}%
{ \end{list} }

\title{ Time-consistent pension policy with minimum guarantee and sustainability constraint\thanks{With the financial
support of Europlace Institute of Finance (EIF) for the project ``Intergenerational risk sharing in pension plans". The authors's research is part of the ANR project DREAMeS (ANR-21-CE46-0002).}}  

\author{Hillairet Caroline,~ \thanks
{\small CREST, UMR CNRS 9194, ENSAE, Institut Polytechnique de  Paris}
\and Kaaka\"i Sarah
\thanks{ \small  Laboratoire Manceau de Math\'ematiques, Institut du Risque et de l'Assurance, Le Mans Université
}\\
 \and Mrad~Mohamed~ \thanks
{\small   LAGA, UMR CNRS 7539,  { Universit\'e Sorbonne Paris Nord}}
}
\date{}

\begin{document}
 \maketitle

  \abstract{   This paper  proposes and investigates an optimal  pair investment/pension policy for  a  pay-as-you-go (PAYG) pension  scheme. The social planner can invest in a buffer fund in order to guarantee a minimal pension amount. The model aims at taking into account complex dynamic phenomena  such as the demographic risk and its evolution over time, the time and age dependence of agents preferences, and financial risks. 
  The preference criterion of the social planner is  modeled by a  consistent dynamic utility defined on a stochastic domain,   which incorporates the heterogeneity of overlapping generations and its evolution over time. The preference criterion and the optimization problem also incorporate sustainability, adequacy and fairness constraints.  
The paper designs and solves the social planner's dynamic decision criterion,  and  computes the optimal investment/pension  policy  in a general framework. A  detailed analysis for the case of dynamic power utilities is provided. 
  }

\bigskip


\section*{Introduction}

{ In many countries,  public pensions are financed by a Pay-as-you-go (PAYG)  scheme, where current pensions are financed by the redistribution of contributions paid by working participants. In order for the pension system to remain sustainable, these contributions should match approximately the pension liabilities. Intergenerational solidarity is thus one of the main pillars of PAYG pension systems.} However, the sustainability of PAYG systems has become a key challenge for policymakers in aging countries, caused  by both the decrease in birth rates and an unprecedented increase in the  life expectancy (see e.g. \cite{aglietta2002demographie}, \cite{VAUPEL2009ageingchallenges}, \cite{masson2012partager}). { For instance  in OECD countries, pension expenditures have increased on average by 1.5 \% of GDP in nearly two decades, and  an
additional  increase of 3.5\% of GDP by 2035 is forecasted (\cite{OECD21}). As a consequence, many governments 
have engaged pension reforms with the aim to reduce pension expenditures. Among them, several countries like  Italy or Sweden have redefined the  pension amount formula from a Defined Benefit PAYG to a Notional Defined Contribution PAYG, with a pension indexation rate which can depend on demographic and/or economic indicators (see \cite{godinez2016optimal}, \cite{alonso2019continuous}).  However, even NDC PAYG pension systems are not guaranteed to be sustainable (see \cite{alonso2017adequacy}, \cite{alonso2019continuous}).}  On the other hand, pension systems should  provide adequate benefits for retirees as well as an acceptable level of fairness between generations, as underlined in \cite{alonso2017adequacy}.  {As an example, only four OECD countries do not guarantee a minimum pension amount (see \cite{OECD21}).}\\
The aim of this paper is to  propose an adaptive decision criterion in order to design an optimal policy that is consistent with both  sustainability and adequacy constraints in PAYG framework. 
{ Our contribution is twofold.\\
Firstly, we introduce a general dynamic framework for  modeling a PAYG pension system,  with a guaranteed minimum annuitized pension amount.}  Pensions are financed by the workers' contribution, and the social planner has also  the flexibility to invest/borrow  from a buffer fund to finance the pension system. The buffer fund allows the social planner to invest in  the financial market. {In addition, the reserves accumulated in the fund allows for demographic (longevity/fertility) shocks to be mitigated over time (\cite{pantelous2008dynamic}).}
An important challenge  is  to convey the complexity of  the problem, by taking into account key phenomena such as the demographic risk and its evolution over time, the time and age dependence of agents' preferences, or financial risks. To the best of our knowledge, these problems have only been tackled either partially or separately in the literature. \\
We adopt a dynamic and continuous time approach, which incorporates the heterogeneity of overlapping generations and the non-stationary evolution of the population over time.  
   A similar flexible population dynamics model is studied in  \cite{alonso2019continuous}, however with deterministic age and time dependent birth and mortality rates, while  we consider  stochastic demographic rates  in our setting. This allows us to take into account stylized fact of the population dynamics, such as uncertain longevity or  dependency ratio increases. In particular, the population age composition is not assumed to be  stationary, as it often the case (see e.g. \cite{DeMenil2006planning}, \cite{cui2011intergenerational}). Furthermore,  we do not restrict ourselves to a two-period models, as in \cite{bohn2001social}  or \cite{Devolder2015} for instance. \\
{Several papers in the literature study deterministic models of PAYG pension systems,  in which the social planner can invest in a buffer fund (see e.g. 
\cite{haberman2002investigation}, \cite{pantelous2008dynamic},  \cite{godinez2016optimal}, \cite{godinez2016finance}, \cite{alonso2019continuous}).} However, the buffer fund is assumed to have a known return. {Besides, when the  optimal pension policy is derived from optimizing solvency indicators, a sustainability constraint is only taken into account at an arbitrary terminal time, and there are no adequacy constraints.} In this paper, the sustainability of the pension scheme is ensured by imposing a pathwise solvency constraint on the buffer fund (see also \cite{di2011pension}  for a similar assumption  in the case of a fully funded pension system). \\
In \cite{gabay2012fair},  the social planner can invest in a complete market, and a minimum pension amount is considered. However, it is obviously impossible  to hedge perfectly the demographic and economic risks through the financial market. Hence, we take into account in this paper the  incompleteness of the financial market. Furthermore, the pension benefit in \cite{gabay2012fair} is a  lumpsum at retirement. {This means that the longevity risk is not taken into account,   contrary to this paper, since we  consider here an adequacy constraint (minimum pension) directly on the pension amount received by pensioners until their death.} \\[1mm]
\indent {Secondly,  this paper provides  contributions on the literature on  forward dynamic utilities with stochastic constraints and endowments.} Indeed, as the representative of past, present and future generations, 
the social planner should aggregate  preferences of all pensioners. 
This aggregation is the key in the fairness criterion as this benevolent social planner aims at  dealing with successive overlapping generations fairly. Thus,  the social planner's decision criterion that appears in the optimization problem's  formulation  is composed of the buffer fund utility and an  aggregated utility which  should capture the heterogeneous preferences of different generations.   {This motivates the use of forward dynamic utilities to  deal with such a complex framework (\cite{Aggregation})}. 
\\
Technically, the problem formulation   is related to the literature on optimal investment and consumption with labor income,  respectively corresponding in our setting to the buffer fund, pensions  and contributions.  
The literature usually states this optimization problem in a backward formulation, see e.g. \cite{he1993labor}, \cite{cuoco1997optimal}, \cite{el1998optimization}, or more recently \cite{mostovyi2020optimal} for a general setting in incomplete markets. However the backward approach has several drawbacks when considering the framework of PAYG pensions. First it does not incorporate any changes in the agents' preferences, or any uncertain evolution of the environment variables. Furthermore, in the context of pensions, fixing a time-horizon is difficult and can lead to optimal choices that depend on the time horizon, inducing  artificial phenomena  such as the  fund liquidation at terminal time (see e.g. \cite{gabay2012fair}, \cite{godinez2016finance}). Finally, the social planner should be able to 
 identify  her preference at any intermediate  time, in order to ensure consistency across time and generations. \\
 Introduced by Musiela and Zariphopoulou \cite{zarD,zar2011}, the framework of dynamic utilities {(also called forward preference processes)} is well suited to solve the issues raised above. Dynamic utilities allow their users to  propose long-term, time-coherent policies adjusted to the information flow, in non-stationary and uncertain environment. {The framework  has been generalized  and extended to many classes of dynamic utilities (see e.g \cite{zar-04}, \cite{MrNek1}, \cite{zar2017}, \cite{mrad21}) and to many applications, including (among others) risk measures (\cite{zariphopoulou2010maturity}), pricing and risk sharing (\cite{anthropelos2014forward}), yield curve modeling (\cite{el2022ramsey}), or competition in fund management (\cite{anthropelos2022competition,dos2022forward}). In a recent preprint, \cite{ng2023optimal} study the optimal investment of an individual worker in a funded defined contribution pension fund. To the best of our knowledge, our paper is the first attempt to study PAYG pension schemes from a dynamic utility point of view. 
 The setting is close to the forward  dynamic utilities of investment and consumption framework  in \cite{MrNek07},  whose  results   are extended here by considering  a given continuous stream of income (contributions) and  stochastic pathwise constraints on the buffer fund (wealth) and pensions (consumption). }\\
The population model, PAYG pension system with sustainability and adequacy constraints, and incomplete financial market  are  introduced in Section \ref{SectionModel}. The social planner's dynamic decision criterion is formulated in Section \ref{SectionDecision}, by introducing dynamic utilities defined on stochastic domains. In particular, we give sufficient conditions on the local characteristics of the dynamic utility for   the utility to be  well defined. The aggregation of  pensioners' preferences is introduced in Section \ref{SubSecBuffUtility}.  The main results of the paper are presented in Section \ref{SectionMainResults}. We first introduce a natural consistency HJB-SPDE derived from the dynamic programming principle, so that the  decision criterion is consistent over time. Under this sufficient condition, the optimal constrained investment and pension policy is derived explicitly in Theorem \ref{ThOpt}. The remainder of the section is dedicated to proving this result. Section \ref{sec:exCRRA} details two examples in the case of dynamic power utilities. 

\section{The model}
 \label{SectionModel}
 All stochastic processes are defined on a standard filtered probability space
$ (\Omega, \mathcal F, {\mathbb F},\mathbb{P})$, where
the filtration ${\bF}=(\cF_t)_{t\ge 0}$ is the natural filtration generated a $n$-dimensional Brownian motion $W$,   and that is  assumed to be right continuous and
complete.

\subsection{Population dynamics}
{We consider a population structured continuously in age and time, characterized by the family $(n(\cdot,a))_{a\geq 0}$ of $\mathbb F$-adapted processes, where $n(t,a)$ is the number of individuals of age $a$ at time $t$. For ease of notation, we consider in this paper a single-sex framework. We only assume that the total population size  remains finite, i.e.: 
\begin{equation*}
N_t : = \int_0^\infty n(t,a) da <\infty, \quad \forall \; t\geq 0 \; \;  \mathbb{P}\text{-a.s.}
\end{equation*}
This general framework allows for the modeling of the stylized facts and \textit{non-stationarity} of an aging population evolution.  In particular, the population evolution is usually characterized by age-specific demographic (mortality and birth) rates. In our general framework, they can be  modeled for instance by   a family of $\mathbb F$-adapted nonnegative processes $(d(\cdot,a))_{a\geq 0}$ and $(b(\cdot,a))_{a \geq 0}$. Time-dependent stochastic birth and mortality rates can describe the uncertain aging of the population, by taking into account phenomena such as the decrease of mortality rates over time (longevity risk), or birth rates declines.  For instance,  stochastic mortality models fitted to real data (see e.g. \cite{CAIRNS2008},\cite{andres2018stmomo}) can be used as inputs for death rates.}

The population dynamics can be derived from these rates. Between a small period of time $[t,t+ dt]$, a (random) proportion $d(t,a)dt$ of individuals of age $a$ die, while individuals of age $a$ give birth to $b(t,a)n(t,a)$ individuals. In this example, the population dynamics is formally described by a partial differential equation with stochastic coefficients, generalizing the standard McKendrick-Von Foerster  equations:
\begin{align}
& (\partial_t + \partial_a) n(t,a) = - d(t,a)n(t,a), \\
& n(t,0) = \int_0^\infty b(t,a)n(t,a)da. 
\end{align}
A detailed analysis of such equations can be found in \cite{KAAKAI201916} (see also e.g. \cite{Hopp75}, \cite{webb1985}).  
In particular, such models  can be easily extended  to include intra-cohort heterogeneity or exogenous population flows.   \\

Workers are assumed to enter the work force at  fixed age $a_e$, and retires at age $a_r$.  Thus, the number of pensioners at time $t$ is 
\begin{equation}
\label{Retirees}
N_t^r=\int_{a_r}^{\infty} n(t,a) da 
\end{equation}social
and the number of workers is 
\begin{equation}
\label{Workers}
N_t^w=\int_{a_e}^{a_r} n(t,a) da 
\end{equation}
In the following, all results can be  straightforwardly extended when the retirement  age $a_r(t)$ (or $a_e(t)$) depends on time. For ease of notations, the time variable is omitted in the following.
\subsection{Pension system}
We define below the aggregate contribution rate process $C_t$ payed by the workers at time $t$, and the aggregate pension rate process $P_t$ received by the pensioners at time $t$.  In what follows, we assume that individuals in a given cohort (of age $a$ at time $t$) receive the same wage and  pension amount.   In the case of heterogeneity inside the cohort, the age-dependent wage and pension processes can be interpreted as  the average amounts over the cohort.
\paragraph{Contribution process} The wage process of any  worker of age $a \in[a_e,a_r]$ at time $t$ is assumed to be an exogenous $\mathbb F$-adapted process $(\mathfrak{e}_t(a))$.  All workers contribute a fixed proportion $\alpha_c \in [0,1]$ of their wage in order to finance pensions of current pensioners.  The aggregate contribution process $(C_t)$ is thus defined by 
\begin{equation}
\label{ContributionProcess}
C_t = \alpha_c   \int_{a_e}^{a_r} \mathfrak{e}_t(a)  n(t,a) da,  \quad \forall \; t\geq 0. 
\end{equation}
If the wage does not depend on the age, then $C_t = \alpha_c \mathfrak{e}_tN_t^w.$
\paragraph{Pension process and adequacy constraint}  In order to achieve a minimum adequacy of retirement incomes, we consider the following pension mechanism: Each pensioner of age $a$ at time $t$ receives a pension amount $p_t(a)$ depending on her age, with
\begin{equation}\label{eqpensionage}
p_{t}(a) = p^{min}_{t}(a) \rho_t, \quad a\geq a_r, \; t\geq 0.
\end{equation}
$p^{min}_{t}(a)$  corresponds to a minimum pension amount guaranteed to pensioners (see examples below) and $(\rho_t)$ represents a global adjustment with respect to the benchmark, which  is the same for all pensioners living at time $t$.  The pension process should thus satisfy the  \textbf{adequacy constraint} \eqref{AdequacyConstraint} below:
\begin{equation}
\label{AdequacyConstraint}
\forall t\geq  0,  \forall a \geq a_r, \;p_{t}(a)  \geq p^{min}_{t}(a)  \text{ a.s., \; or equivalently } \rho_t \geq 1  \text{ a.s.}
\end{equation}
Finally, the total pension amount is 
\begin{equation}
 \label{Pensionprocess}
P_t  = \rho_t P_{t}^{min}, \text{ with } P_t^{min } = \int_{a_r}^\infty p^{min}_{t}(a) n(t,a)da. 
\end{equation}
\noindent In the following, two examples are studied as an application of the general results:  
\begin{Example}\label{Ex1}{\normalfont 
When $p^{min}_{t}(a) \equiv  p_t^{min}$, all pensioners at time $t$ receive the same pension amount $ p_t^{min} \rho_t$ and the minimal aggregate pension amount is $P^{min}_{t}=N_{t}^rp^{min}_{t}$.
 For example, $p_t^{min}$ can be indexed on current wages, contributions, or any other indicator.}
\end{Example}
\begin{Example}
\label{Ex2}{\normalfont 
 A more realistic choice for $p^{min}$ is to take  a base pension computed  at retirement date, multiplied by an indexation factor: 
\begin{equation}
\label{Eqpref_ex2}
p^{min}_{t}(a) = p_{ret}(a_r + t-a) \; e^{\int_{a_r+t-a}^t\lambda_s ds}.
\end{equation}
$\bullet$  $p_{ret}(s)$ corresponds to the base pension amount received by an individual retiring (i.e. of age $a_r$) at time $s$. Then $p_{ret}(a_r + t-a)$ is the pension amount received at retirement time by an individual of age $a$ at time $t$.   For instance, $p_{ret}(s)$ can be  a proportion $\alpha_p$  of the average yearly income of an individual retiring at time $s$. Then
\begin{equation}\label{defPastcontrib}
p_{rep}(s) =  \alpha_p \frac{c_{s}(a_r )}{a_r - a_e}, \text{ with } c_{t}(a)=  \int_{t - a + a_e}^{t- a+ a_r}   e^{\int_u^t r_s ds } \mathfrak{e}_u(a +u-t)  du.
\end{equation}
$c_{t}(a)$ is the present value of  wages earned by an individual's of age $a$ at $t$.\\ 
One can also take instead  the average income over the last  $h$ years before retirement, then $p_{ret}(s) =  \alpha_p \frac{\int_{s-h}^{s}   e^{\int_u^s  r_v dv }\mathfrak{e}_u(a +u-t) du}{h}$.\\
$\bullet$ $(\lambda_t)$ is the indexation rate adjusting pension benefits. The indexation rate  takes into account changes in prices or wages, using for example the Consumption Price Index. It can be used  to maintain the sustainability of the pension system is case of a demographic shock, 
such as in  \cite{alonso2017adequacy}. }
\end{Example}
\subsection{Buffer fund}
{In a ``pure" PAYG scheme,  the aggregate pension amount $P_t$  to be paid at time $t$  to current pensioners is only financed  by the  contributions  of current workers $C_t$, inducing the budget constraint: 
\begin{equation}
\label{budgetconstraintPAYG}
P_t = C_t , \quad \forall t\geq 0.
\end{equation}
The budget constraint \eqref{budgetconstraintPAYG} ensures the sustainability/self-financing of the PAYG system.  When the  contribution rate $\alpha_c$ is fixed, the adjustment coefficient $\rho_t$ is  automatically determined by: 
\begin{equation*}
\rho_t = \frac{C_t}{P_t^{min}} = \alpha_c \frac{\int_{a_e}^{a_r} \mathfrak{e}_t(a)  n(t,a) da}{\int_{a_r}^\infty p^{min}_{t}(a) n(t,a)da}, 
\end{equation*} 
and does not necessarily satisfy the adequacy constraint \eqref{AdequacyConstraint}.   The adequacy constraint yields a liquidity issue for the PAYG system, since the minimum pension amount $P^{min}_t$ is not  covered by the workers contributions when $P^{min}_t  \geq C_t$.\\
In particular, the retirement income can be drastically reduced when  there is a decrease in the working population $N_t^w = \int_{a_e}^{a_r} n(t,a) da$, resulting in a decrease in the social planner's stream of income, or an increase of the retired population
$N_t^r = \int_{a_r}^{\infty} n(t,a) da$, resulting in an increase of the pension amout $P_t$. This  ``demographic risk" can be induced by a fertility shock or an decrease over time of death rates for older ages, as currently observed in most countries  (see e.g. \cite{barrieu2012understanding} on the longevity risk).  Pensioners  also bear all the economic risk  through the wage process $(\mathfrak{e}_t)$.
We thus consider a system where the social planner can invest (or borrow) at each time the amount $C_t - P_t$ in (from) a buffer fund, hence sharing the demographic risk between the different generations.}
\paragraph{Financial market}  As the demographic (and economic) risks are  obviously not fully transferable to financial markets, we consider   an incomplete It\^o  market,  with short rate $(r_t)$ and $d$  risky  assets ($d \leq n$). The price process $S = (S^i)_{i=1,\dots,d}$ of the risky assets is defined by 
\begin{equation}
dS^i_t = S^i_t (\mu^i_t dt + \sum_{j=1}^n \sigma^{i,j}_t dW^j_t), \quad \quad  i=1,\cdots,d.
\end{equation}
 The $d \times n$  volatility matrix   $(\sigma_t)$  is assumed of full rank (that is $ (\sigma_t{}^{tr}\sigma_t)$ is invertible, where ${}^{tr}\sigma_t$ is the transposed matrix).  
The  $d$-dimensional risk premium vector is denoted by $(\eta_t)$, where  $\eta_t={}^{tr}\sigma_t. (\sigma_t{}^{tr}\sigma_t)^ {-1 } (\mu_t -r_t \bf{1}_d)$.  
 $\eta_t $  is in  the   vector subspace $ \mathcal{R}_t = {}^{tr} \sigma_t( \mathbb R^d)$   of $\mathbb R^n$. We assume that $\int_0^T( |r_t| + \|\eta_t\|^2)dt<\infty$, for any $T>0$, $\mathbb P$.a.s. \\
In what follows, for any $ \mathbb R^n$-valued  stochastic process $(X_t)$,  we denote by $X^\mathcal{R}$ the process such that  for all $t$,  $X_t^\mathcal{R}$ is the orthogonal projection of $X_t$ onto $\mathcal{R}_t$. Similarly $X_t^\perp$  denotes  the   orthogonal projection  of $X_t$ onto the orthogonal vector subspace  $\mathcal{R}_t^\perp$.\\
{Note that thanks to the weak assumptions made on the demographic,  wage and financial processes, any kind of dependency structure with the financial market  can be taken into account in the modeling.}
\paragraph{Investment in the buffer fund} On this (incomplete) market,  the  social planner manages a buffer fund  that aims to absorb demographic and economic shocks. 
The dynamics of the fund can be considered as the financial wealth of a single agent with  a labor income (or endowment) $C_t = \alpha_c   \int_{a_e}^{a_r} \mathfrak{e}_t(a)  n(t,a) da$, and consumption process $P_t$. 
The amount of money invested in the risky assets at time $t$ is denoted by the  $d$-dimensional vector $\phi_t$. Then the 
self-financing dynamics for the wealth $F$ of the buffer fund with contribution $C_t$ and pension $P_t$ is 
$$d F_t=  F_t  r_t dt + (C_t-P_t) dt +  \phi_t \cdot \sigma_t . ( dW_t + \eta_t dt).
$$
In order to ensure the sustainability of the pension system, a maximum  amount of debt  is imposed to the buffer fund :
\begin{equation}
\label{ConstraintK2}
\text{(Sustainability constraint)} \quad \forall t \geq 0,  \;   F_t  \geq  \mathfrak{K}_t , a.s., 
\end{equation} 
$ \mathfrak{K}_t$   is an adapted predictable process, that can be negative, and $( - \mathfrak{K}_t)$   corresponds to the  social planner maximum amount of debt  at time $t$.
For instance, $(-\mathfrak{K}_t)$ can represent a proportion of the Gross Domestic Product (GDP), while the no borrowing constraint corresponds to $ \mathfrak{K}_t \equiv 0$. We assume throughout the paper that $\mathfrak{K}$ is an Itô process:
\begin{equation}
\label{EqDynK}
d\mathfrak{K}_t = \mu_t^{\mathfrak{K}} dt + \delta^{\mathfrak{K}}_t \cdot dW_t. 
\end{equation}
\paragraph{Social planner strategy} The strategy of the social planner (pensions and investment) is parametrized by $(\pi, \rho)$ where  $P_t = \rho_t P^{min}_t$ and $\pi_t := {}^{tr}\sigma_t  \phi_ t \in \mathcal{R}_t$. The dynamic dynamic of the buffer with strategy $(\pi, \rho)$ is thus given by: 
\begin{equation}
\label{EqFond}
d F^{\pi, \rho}_t=  F^{\pi, \rho}_t  r_t dt + (C_t-\rho_t P^{min}_t) dt +  \pi_t  \cdot  ( dW_t + \eta_t dt).
\end{equation}
Observe that the previous equation is not in a multiplicative form since the value of the fund  $F$ can be negative.
\begin{Definition}[Admissible strategy]
An $\mathbb{F}$-adapted strategy  $(\pi,\rho)$ is said to be admissible  if and only if
 \begin{itemize}
\item $\pi_t \in \cal R_t$,   $\; \; \forall \; t\geq 0 \;  \; \P\text{-a.s.} $
\item  $\int_0^t \left( | C_s - \rho_s P^{min}_s | + \| \pi_s\|^2 \right)ds <\infty,   \; \; \forall \; t\geq 0 \;  \; \P\text{-a.s.} $
\item $\rho_t \geq 1$ (adequacy) and   $ F_t^{\pi,\rho} \geq  \mathfrak{K}_t$ (sustainability) $\; \; \forall \; t\geq 0 \;  \; \P\text{-a.s.} $
 \end{itemize}
 The set of all  admissible strategies   $(\pi,\rho)$ is  denoted $\A$.
\end{Definition}
Lastly, we introduce  the class of the so-called state price density processes (taking into account the discount factor), also called discounted  pricing kernels. 
The discounted pricing kernels $Y$  are characterized by the property that for any admissible strategy $(\pi,\rho)$, the current wealth plus the aggregate pension minus the  contributions, all discounted by $Y$(that is  $Y_t F^{\pi, \rho}_t +\int_0^t  (\rho_s P^{min}_s-C_s) Y_s ds  $) is a  local martingale (see  \cite{el2022ramsey} for  financial and economic viewpoints on the discounted  pricing kernels).
Discounted  pricing kernels are  positive Itô-semimartingale with   the following dynamics characterized by an orthogonal volatility component $ \nu_t \in \mathcal{R}_t^\perp$.
\begin{Definition}(State price density process/discounted  pricing kernel).
 A positive Itô semimartingale $Y^\nu$   is called an admissible state price density process  (or discounted  pricing kernel) if it has  the differential decomposition 
 \begin{equation}\label{Ydyn}
 dY^\nu= Y^\nu[ - r_t dt + (\nu_t - \eta_t) \cdot dW_t], \quad \nu_t \in \mathcal{R}_t^\perp , \; Y_0^\nu = Y_0.
 \end{equation}
\end{Definition}

\section{The social planner's dynamic decision criterion}
\label{SectionDecision}
The pension allocation and the investment in the buffer  fund is decided by a social planner. Her decision are taken upon a preference criterion that should take into account  present and future  generations. Since the social planner  has to  aggregate the  heterogeneous preferences of the different cohorts, her utility criteria is necessarily complex. \\
Moreover, due to the long-term characteristics of pension schemes,  the decision criteria should be adapted  to the non-stationary  demographic, economic and financial  environment in order to provide a consistent strategy in the long run. 
For both of these aspects, it has been shown that dynamic utility provides a flexible framework   to handle this heterogeneity and to propose  long-term policies  coherent in time (\cite{Aggregation},\cite{MrNek07}). 

\subsection{Definition of Dynamic Utilities}
\label{DefDynUtility}
Dynamic utilities extend the standard notion of deterministic utilities,  by allowing the 
 utility criterion to be  dynamically
adjusted  to  
the available information represented by the filtration $(\cF_t)_{t \ge 0}$. 
A dynamic utility $U$  is  a then collection of random utility functions $\{U(t,\omega, z)\}$ whose   temporal evolution is  ``updated"  in accordance with the new information $(\cF_t)$, starting from an initial utility function $u(z)=U(0,z)$ (which is deterministic if $\cF_0$ is trivial). Throughout the paper, we adopt the convention of  small letters  for deterministic utilities and capital letters  for stochastic utilities.  The specificity of this paper lies in the sustainability and adequacy constraints for the buffer fund and the pension process, which will be taken into account in the definition domain of dynamic utilities. In the following, we give a precise definition of this extension of dynamic utilities on stochastic   domains. 
\begin{Definition}[Dynamic Utility on  stochastic  domain]${}$\\
\label{DynamicUtility}
A dynamic  utility defined on a random domain $U = (U(t,z,\omega))$  is a collection of random utility functions defined on  a stochastic domain $ \mathcal D_U :=\{ (t,z,\omega),  z \geq X_t(\omega) \}$ such that \\ 
\rmi  $(X_t)$ is an $\F$-adapted  process.\\
\rmii  For all $t \geq 0$, for all $z \geq X_t$,  $U (t, z)$ is $\mathcal F_t$-adapted. \\
\rmiii There exists $N \in \mathcal F$ with $\P(N)=0$, such that  for any $\omega\in N^c$,   and for any $t \geq 0$,  the functions $z \in [ X_t(\omega), \infty[  \mapsto U(t,z,\omega)$  are  nonnegative,  strictly concave increasing  functions  of class $\C^{2}$ on $]X_t(\omega), \infty[$.\\
\rmiv Inada conditions: 
$\quad \lim\limits_{z \rightarrow X_t(\omega)} U_z(t, z,\omega)= + \infty  \text{ and }  \lim\limits_{z \rightarrow + \infty} U_z(t, z,\omega)= 0, \quad \forall t\geq 0\;  \mathbb{P} \text{-a.s.}$
\end{Definition} 

\noindent As for a standard utility function, the   risk aversion coefficient ${\rm R_A}(U)$ is measured by the ratio ${\rm R_A}(U)(t,z)=-U_{zz}(t,z)/U_z(t,z)$  and the relative risk aversion by ${\rm R^r_A}(U)(t,z)=z\,{\rm R_A}(U)(t,z)$. Note that for dynamic utilities, ${\rm R_A}$ and ${\rm R^r_A}$ are random coefficients.
\begin{Remark}\label{RemarkInada} Note that the Inada condition in $X_t$ implies that  the absolute risk aversion coefficient explodes in $X_t$: 
\begin{equation}\label{consInada}
\lim_{z \rightarrow X_t(\omega)}  R_A(U)(t,z) = \lim_{ z \rightarrow X_t(\omega)}  - \frac{{U}_{zz}(t,z)}{ U_z(t,z)  }=+ \infty, \;  \mathbb{P} \text{-a.s }.
\end{equation}
\end{Remark}
\noindent Indeed, denoting $f(t,z):= -\frac{{U}_{zz}(t,z)}{ U_z(t,z)  }$,  we have that $U_z(t,z) = U_z(t,z_0) \exp(\int_z^{z_0} f(t,u)du)$  for  $z_0>z$  greater than $X_t$.
Then, for a fixed $z_{0}$,  the Inada condition  $\lim\limits_{z \rightarrow X_t(\omega)} U_z(t, z)= + \infty$  implies   $\lim\limits_{ z \rightarrow  X_t(\omega)}  \int_{z}^{z_{0}} f(t,u)du = +\infty$ and therefore $\lim\limits_{z \rightarrow X_t(\omega)} f(t,z)=+ \infty .$\\

\noindent  The Fenchel-Legendre convex conjugate of a dynamic utility $U$  defined on a stochastic domain $ \mathcal D_U :=\{ (\omega, t,z),  z \geq X_t(\omega) \}$ is denoted by $\tU$,  where   $\tU$  satisfies  
$$\tU(t,y)=\sup_{z \geq X_t(\omega) }\big(U(t,z)-yz), \quad y \in \mathbb R^+ .$$
 In particular, $\tU(t,y)\geq U(t,z)-yz$ and the maximum is attained at  $U_z(t,z)=y$. Note  that  $\tU$  is defined on $\R^+ \times \R^{+}$ thanks to Inada conditions on $U$.  
 $\tU$  is  { twice continuously} $y$-differentiable, strictly convex, strictly decreasing.
 Moreover, the marginal  utility $U_z$ verifies $U_z^{-1}(t,y)=-\tU_y(t,y)$;  $\tU(t,y)=U\big(t,-\tU(t,y)\big) +\tU_y(t,y)\, y$, and $U(t,z)=\tU\big( t,U_z(t,z)\big)+z U_z(t,z)$.\\

\noindent {As this paper has an economic motivation, it is  natural  to study the  particular example of CRRA utility functions, and to consider their dynamic versions.}

\paragraph{Dynamic CRRA utilities}  Deterministic Constant Relative Risk Aversion (CRRA) utilities  is  the standard framework  in the economic literature, as explained in \cite{wakker2008explaining}. They belong to the class of 
deterministic Hyperbolic Absolute Risk Aversion (HARA) utility (see e.g. the seminal paper of Merton \cite{merton1975optimum}, or Kingston  and Thorp \cite{kingston2005annuitization}) 
 characterized by an  hyperbolic  absolute risk aversion coefficient:
 \begin{equation}\label{HARA}
 {\rm R_A}(u)(z)=-\frac{u_{zz}(z)}{u_z(z)} =\frac{1}{az+b},\quad  \text{with } a  >1  \text{ and } b \in \mathbb R.
 \end{equation}
Integrating  \eqref{HARA} gives that 
$u(z) = \frac{ (z+b/a)^{1-\theta}}{1-\theta}$, defined for   $z > -b/a$, and  where $\theta=1/a$. The case $b=0$ corresponds to CRRA utility (also called power utilities) with constant relative risk aversion $\theta$.  
 Hereafter, we extend the notion of deterministic CRRA utilities  to dynamic ones,  still with deterministic  $\theta$.  {As it is well known, CRRA dynamic utilities  are necessarily time-separable.\footnote{Indeed, assuming 
$-\frac{zU^{(\theta)}_{zz}(t,z)}{U^{(\theta)}_{z}(t,z)}= \theta  \in ]0,1[$ 
implies  $\log(U^{(\theta)}_{z}(t,z))=- \theta \log(z)+c_{t}$ for $c$ an $\mathbb F$-adapted process, thus $U^{(\theta)}(t,z)=\frac{z^{1-\theta}}{1-\theta} Z_{t}$  with $Z=e^c$.}
}

 \begin{Definition}
 {\it Dynamic CRRA utility}  $ U^{(\theta)}(t,z)$, with  $\theta \in ]0,1[$, are defined by 
\begin{equation}\label{CRRAF}
U^{(\theta)}(t,z) := Z^{u}_t  \frac{(z - X_t)^{1-\theta}}{1-\theta},  \;  \mbox{ for  }  z \geq   X_t
\end{equation}
where  $X_t$ is a stochastic process and $Z^{u}_t$ is a positive  stochastic coefficient reflecting the random evolution of the time preferences.
\end{Definition}
\noindent For instance $X_t$ can represent a borrowing constraint of the social planner. 
 Dynamic  CRRA   utilities satisfy  Inada conditions :    $$\lim\limits_{z\to \infty} U^{(\theta)}_z(t,z)=0 \mbox{  and }  \lim\limits_{z\to X_t(\omega)} U^{(\theta)}_z(t,z)=\infty, \quad \forall t, \; \mathbb P \, a.s..$$
 
\noindent CRRA dynamic utilities are also studied in a stochastic factor framework in \cite{zar2017}, where the time factor $Z^{u}$ is  solution of an ergodic BSDE.   In the dynamic setting, the time factor $Z^u$ is more general than a  time preference coefficient, in particular its diffusion coefficient  will impact the optimal strategies (cf. Section \ref{sec:exCRRA}).  Indeed the dynamics of  $Z^u$ is determined by  the consistency condition    (in the sense of Definition \ref{def:consistent system}).   \\
 An advantage of  CRRA  dynamic utilities is that they are particularly tractable and lead to explicit solutions. 
 Other classes of dynamic utilities may  also be considered, such as   exponential  (see \cite{zarD}),   time-monotone (see \cite{zar2011}),  or mixture   of dynamic utilities (see \cite{mrad21}). Nevertheless, more general classes are less tractable and  lead to less explicit results (especially concerning the consistency condition), one could then rely on numerical tools to exhibit the optimal solutions.\\

\subsection{ Buffer fund and pensioners dynamic utilities}
\label{SubSecBuffUtility}
{The aim of the social planner is to optimize the pension of present and future pensioners, and thus has to take into account an aggregation of cohorts. The buffer fund $F$ plays the role of a risk sharing mechanism between generations. The greater are the part of the contribution $C$ invested in the fund, the more reserve are available to pay future pensions and to face unpredictable future demographic and economic risks. On the other hand, an adequate level of pension should be paid to current retirees. This tradeoff is summarized by  the social planner's preference process  defined as $U(t,F_t) + \int_0^t V(s,\rho_s)ds$}.  $U$ is the buffer fund dynamic utility ({representing future generations}), and $V$ is the aggregate dynamic  utility of the pensioners. \\
The sustainability constraint  \eqref{ConstraintK2} and  the adequacy constraint \eqref{AdequacyConstraint}  are translated into stochastic domains for both $U$ and $V$. 
The supermartingale property induced by the dynamic programming principle  will be  translated into condition on local  characteristics of $U$. This explains why we assume stronger  regularity conditions on $U$ than on  $V$. 
\begin{Definition}[Buffer fund utility]\label{defU}
The buffer fund utility $U$ is a dynamic utility with domain $\mathcal{D}_U =\{ (\omega,t,z), \; z\geq \mathfrak{K}_t \}$, where  $(\mathfrak{K}_t)$ is the sustainability bound satisfying \eqref{EqDynK} and with $U(t,\cdot)$ of class $\mathcal{C}^{3,\delta}$,  that is of class $\mathcal{C}^3$ with ${U_{zzz}}$ $\delta$-H\"older, $\delta \in ]0,1[$. 
\end{Definition}

\noindent   The preference process  of a pensioner of age $a$ at time $t$,  is defined by  $\bar{v}(t,a, p_t(a))$, with $p_t(a)$ then pension amount, and  $\bar{v}$ a dynamic utility depending of the pensioner's age and  taking into account uncertain future changes in the pensioners' preferences.  \\
The social planner aggregates preferences of all pensioners to obtain the aggregated pensioners' dynamic utility defined by  $\int^\infty_{a_r}  \bar v(t, a,  p_{t}(a)) \omega_{t}(a) n(t,a)da$, with $\omega_{t}(a)$ the weight given to  a pensioner of age $a$ at time $t$. For instance, the social planner can take into account the actuarial fairness by giving more weight to pensioners who contributed more. \\
Recalling that $p_{t}(a) = p^{min}_{t}(a) \rho_t$,  we can define   a weighted dynamic utility $v$ applied to  $\rho_t$, with 
\begin{equation}\label{eq:vweighted}
v(t,a,\rho) =  \omega_{t}(a) \bar v(t, a,  p^{min}_{t}(a)\rho).
\end{equation}
In the following, we refer   to $v$ as a pensioner dynamics utility from the viewpoint of the social planner.  
\begin{Definition}[Pensioners' utility]
\label{hypV}
We assume that a pensioner at time $t$ and of age $a$ has the dynamic utility $v(t,a,\cdot)$, defined as in \eqref{eq:vweighted} on a domain $[\underline{\rho}_t,+\infty[$, with  $\underline{\rho}_t \leq 1$. The aggregated utility of pensioners is defined by: 
\begin{equation}\label{defaggregateV}
V(t,\rho)= \int^\infty_{a_r}  v(t, a,  \rho) n(t,a) da.
\end{equation}
 \end{Definition}
\noindent {\bf Example} \ref{Ex1} We come back to  the case where each pensioner receives the same pension amount  $p_t=\rho_t p^{min}_t \geq p^{min}_t$ and has the same dynamic utility $\bar v(t,p)$. In order to take into account actuarial fairness, the  social planner can  aggregate preferences of all living pensioners by weighting their utility  by  their past contributions  $\alpha_c c_{t}(a)$ (in the same spirit of \cite{gabay2012fair}), introduced in \eqref{defPastcontrib}. Then,
\begin{equation*}
v(t,a, \rho) = \alpha_c c_{t}(a)\bar v(t, \rho \, p^{min}_t), 
\end{equation*}
and the aggregate utility from pensions, taking into account the weight of each pensioners' generation, is
\begin{equation}
\label{eqVexample1}
V(t,\rho) =\int^\infty_{a_r} \bar v(t,\rho \, p^{min}_t )\alpha_c c_{t}(a)n(t,a) da=  \omega^r_t \bar v(t, \rho \, p^{min}_t),
\end{equation}
with the  total weight of all pensioners living at time $t$ given by
\begin{equation}\label{eqpoidsexample1}
\omega^r_t = \alpha_c   \int_{a_r}^{\infty}  c_{t}(a) n(t,a) da .
\end{equation}

\noindent {\bf Example} \ref{Ex2}: In the  second   example of age-dependent pension, 
$p_{t}(a) = p^{min}_{t}(a) \rho_t$ with \\
$p^{min}_{t}(a) = p_{ret}(a_r + t-a) e^{\int_{a_r+t-a}^t\lambda_s ds}$.%
 We further assume that the pensioners utility $\bar v(t,a,p)$ may depend on their age, 
so that  $v(s,a,\rho)  =\bar v(s,a, \rho \, p^{min}_{s}(a))$.  Then the  actuarial fairness criteria is taken into account via the  choice of initial pension amount $p_{ret}(s)$, which increases with the past contribution of the individual retiring at time $s$.  Thus, the social planner may aggregate the pensioners' utility without using any correcting weight. 
Then 
\begin{equation}\label{EqVexample2}
V(t,\rho) = \int^\infty_{a_r} \bar v(t, a, \rho \, p_{ret}(a_r + t-a)e^{\int_{a_r+t-a}^t\lambda_u du} ) \, n(t,a) da
\end{equation}
We introduce the  following assumption on  the   rates of  increase of the marginal dual utilities   $\tilde{V}_\rho $ and $\tilde{U}_z$,  that  will play a role in proving the existence of an admissible optimal portfolio in Theorem \ref{ThOpt}.
\begin{Assumption}\label{hypVbis}
 We assume the existence of a locally integrable  process $B$ such that
\begin{equation}
\label{HypTildeV}
| \tilde{V}_\rho(t,z) - \tilde{V}_\rho (t, z')| \leq B_t \, |\tilde{U}_z(t,z) - \tilde{U}_z(t,z')|,  \; \; \mbox{ for } t \geq 0, \; z>0, \;  z'>0.
\end{equation} 
\end{Assumption}
\noindent Note that  the relative weight of the buffer fund utility $U$ with respect to the pensioners' utility $V$ will play a decisive role. Informally, if the utility of the fund $U(t,F_t)$ is small compared to the pensioners' aggregate utility $V(t,\rho_t)$, the  social planner takes more risks regarding the pension system sustainability (see Section \ref{sec:exCRRA} for more details in the CRRA framework).

\subsection{Consistency Property  }

\textbf{Social planner optimization problem} 
 The social planner has to manage a tradeoff between the pension payed to the pensioners and the fund that constitutes reserves for the future generations, among all the admissible strategies satisfying the sustainability and adequacy. In the usual  setting, the optimization program is posed backward. It is  formulated on a given horizon $T_H$, and is written at time $t=0$ as follows (given $F_0=x$):\\[-2mm]
\begin{equation} \label{optequilib}
\mathcal{U}(0,x):=\underset{(\pi,\rho) \in \A }{\sup} \, \mathbb{E} \big( u( T_H,F^{\pi,\rho}_{T_H}) +\int_0^{T_H} V(t,\rho_t)dt   \big). 
\end{equation}
 In the  backward formulation, the utilities $u$  of terminal wealth (at $T_H$) and $V$ of pension rate are given.  In our context of intergenerational risk-sharing for pensions, fixing a (long-term) time-horizon $T_H$ and even more a utility function $u(T_H,.)$ seems artificial. Extending  the optimization program and the optimal strategy  to a horizon larger to $T_H$, in a time-consistent way, is also a difficult issue.  In order to ensure consistency across time and generations, the social planner should be able to 
 identify which ``terminal" criterion $U(T,.)$ should be considered at any intermediate date $T \leq T_H$, while still leading to the same optimal strategy and   the same value  $\mathcal{U}(0,x)$, that is satisfying \\
 [3mm]
\hspace*{2cm} for any $T \leq T_H$, $\;  \mathcal{U}(0,x)=\underset{(\pi,\rho) \in \A }{\sup} \, \mathbb{E} \big( {U}( T,F^{\pi,\rho}_{T}) +\int_0^{T} V(t,\rho_t)dt   \big). $\\
 [1mm]
 Under regularity assumptions, this criterion is  given by  the ``value function" $\mathcal{U}(T,z)$  given the wealth   $F_T=z$  at time $T$ 
\begin{eqnarray}\label{pbopticlassic}
\mathcal{U}(T,z)=\underset{(\pi,\rho) \in \A }{``\sup"} \, \mathbb{E}  \Big( u(T_H, F^{\pi,\rho}_{T_H}(T,z)) + \int_T^{T_H}
V(s,{\rho_s}) ds  | F_T=z  \Big), \>a.s..
\end{eqnarray} 
This time-consistency translates into a martingale property
of the  preference  process\\ $\mathcal{U}(t,F_t^{(\pi,\rho)}) +\int_0^t V(s,\rho_s) ds$ along  the optimal  strategy. This property, known as  the dynamic programming principle,  is the main tool in the theory of stochastic control, see Davis \cite{DavisLN} or  El Karoui \cite{EKStFlour}.   In this backward setting,  $\mathcal{U}(T_H,.)=u(T_H, .)$ is given, and the unknown is the optimal strategy $(F^*,\rho^*)$ as well as $\mathcal{U}(t,.)$, also called ''indirect''  utility,  possibly stochastic. 
Nevertheless, $\mathcal U$  is difficult to compute (even if  $u(T_H,.)$  is  given  as a   simple deterministic  function), and  it is  even  not  trivial to prove that  $\mathcal U$ defined by \eqref{pbopticlassic} is indeed concave.\\
In the  forward setting,  there is no intrinsic time-horizon $T_H$ and this is the initial utility $\mathcal{U}(0,.)$ which is given. 
 This means that forward utilities differ from backward utilities  by their boundary conditions, both  satisfying a dynamic programming principle, also called  consistency given the constraints set $\A$.
\paragraph{Consistency and optimal strategy}
The satisfaction provided by an admissible  strategy $(\pi,\rho) \in \A$  is measured  by the dynamic  criterion $U(t,F^{\pi,\rho}_t) + \int_0^t V(s,\rho_s) ds$.
that is assumed to satisfy a dynamic programming principle.

\begin{Definition}[Consistent dynamic  utility]\label{def:consistent system}
Let $( U,  V)$ be a dynamic utility  system with admissible strategies set  $\A$. The utility  system $( U,  V)$ is said to be {consistent}, if \\
\rmi For any admissible strategies   $(\pi,\rho) \in \A$,  the preference process 
$(U(t,F^{\pi,\rho}_t) + \int_0^t V(s,\rho_s) ds)$    is a non-negative supermartingale.\\
\rmii There exists an {\rm optimal } strategy  $(\pi^*,\rho^*) \in \A$,  binding the constraints, in the sense that
the optimal  preference process   $(U(t,F^{\pi^*,\rho^*}_t) + \int_0^t V(s,\rho^*_s) ds)$    is a martingale.
 \end{Definition}
 \noindent  Under regularity assumptions, the value function  $(\cal U(t,z), V(t,\rho))$ of the classical optimization problem  is an example of   
 consistent utility system,  defined from its terminal condition $\cal U(T_H,z)=u(z)$ (see  \cite{MrNek07} and \cite{el2022ramsey} for a general discussion between the forward and the backward viewpoints of  utility functions).

\subsection{ Semimartingale  dynamic  utility}\label{SDU}
In order to study the preference process $(U(t,F^{\pi,\rho}_t) + \int_0^t V(s,\rho_s) ds)$, we  assume  in the following that the  buffer fund dynamic utility $U$ defined in \ref{defU} is  an Itô random field:
\begin{eqnarray}\label{rf decomposition}
dU(t,z)=\beta(t,z)dt +\gamma(t,z).dW_t, \quad  z\geq \mathfrak{K}_t,
\end{eqnarray}
with  $\beta(t,z)$ the drift random field   and $\gamma(t,z)$  the  multivariate  diffusion random field. 
Since the domain of  the buffer fund utility $U$ is time varying, its dynamics is defined more precisely by introducing the shifted utility $\bar U$ with fixed domain $\R^+\times \R^+$: 
\begin{equation}
\bar U(t,z) := U(t, z+\mathfrak{K}_t), \quad \forall (t,z) \in \R^+\times \R^+, 
\end{equation}
and with  local characteristics  denoted by $(\bar \beta,\bar \gamma)$. 
Obviously $U$ is a dynamic utility on $\mathcal{D}_U$ if and only if $\bar U$ is a dynamic utility on $\R^+ \times \R^+$.  These  semimartingale dynamic utilities have been studied in details in \cite{MrNek1}. An important  part is the  connection between the regularity of the dynamic  utility  $\bar U$ and that of its local characteristics $(\bar \beta,\bar \gamma)$. If $ \bar U$ is of class $\C^{2,\delta}$ then its characteristics $\bar \beta$ and $\bar \gamma$ are of class $\C^{2,\eps}$ for all $0<\eps<\delta$. Conversely if $\bar \beta$ and $\bar \gamma$ are in the class $\C^{2,\delta}$ then $U$ is in $\C^{2,\eps}$ for all $0<\eps<\delta$ (\cite{Kunita:01}, \cite{MrNek1}). Another part relies on determining  conditions on $(\bar \beta,\bar \gamma)$  to ensure  that $\bar U$ is a dynamic utility.  Indeed, in  the absence of general comparison results for stochastic integrals, it is not straightforward to obtain conditions on the local characteristics $(\bar \beta,\bar \gamma)$  such that the process $\bar U(t,z)=\bar U(t,0) + \int_0^t \bar \beta(t,s)ds +\bar \gamma(s,z).dW_s$  is increasing and concave. 
 The sufficient assumptions recalled below are useful. 
\begin{Assumption} 
\label{HypBarU}
Let $\bar U(t,z) = U(t,z+\mathfrak{K}_t)$ the shifted buffer fund utility.   We assume that there exists random  bounds   $(B^1_t)$ and {$(\zeta_t)$} such that a.s. $\int_{0}^TB_{t}^1dt<\infty$ and {  $\int_{0}^T\zeta_{t}^2dt<\infty$} for any $T$, and  such that for any $z > 0$: 
\begin{numcases}{}
 |\bar \beta_{z}(t,z) | \leq B^1_t \, \bar U_z(t,z), \quad   |\bar \beta_{zz}(t,z) | \leq B^1_t \, | \bar U_{zz}(t, z)   |  \label{boundenzero0} \\
 \|\bar \gamma_{z}(t,z)\|   \leq \zeta_t \, \bar U_{z}(t,z), \quad  \|\bar \gamma_{zz}(t,z)  \|  \leq \zeta_t \,| \bar U_{zz}(t,z )|\label{boundenzero}
 \end{numcases}
\end{Assumption}
\noindent Under  Assumption \ref{HypBarU}, $\bar U$ is  a well-defined dynamic utility by  Corollary 1.3 in \cite{MrNek1}. 

\paragraph{It\^{o}-Ventzel's formula }  The link between the local characteristics $(\beta,\gamma)$ of $U$ and $(\bar \beta, \bar \gamma)$ of $\bar U$ are deduced from the  It\^o-Ventzel formula,  which is a generalization of the  It\^o formula in the case where the function is  itself a random field. \\
The  It\^o-Ventzel formula  gives the decomposition of a compound random field $U(t,X_t)$ for $U(t,z)=u(0,z)+\int_0^t \beta(s,z)ds+\int_0^t \gamma(s,z).dW_s$ regular enough (of class $\C^{2,\delta}$, with $\delta\in ]0,1[$) and  any It\^o  semimartingale $X$. This decomposition is  the sum of three terms: the first one is the "differential in $t$" of $U$, the second one is the classic It\^o's formula (without differentiation in time) and the third one is the infinitesimal covariation between the martingale part of $U_z$ and the martingale part of $X$, all these terms being taken in $X_t$. 
\begin{eqnarray}\label{ItoVentzel}
dU(t,X_t) &=& \big(\beta(t,X_t)\,dt+\gamma(t,X_t).dW_t\big)\\
&+&\big(U_z(t,X_t) dX_t+\demi U_{zz}(t,X_t)d<X,X>_t\big)+\big(<\gamma_z(t,X_t) \cdot dW_t, dX_t>\big).\nonumber
\end{eqnarray}
Applying the Itô-Ventzell formula to $\bar{U}(t, z- \mathfrak{K}_t)$ yields the following result. 
\begin{Proposition} \label{eq:Lipschitzcarac} 
Recall that the bound $\mathfrak{K}_t $ is an Itô process with dynamics $d\mathfrak{K}_t =  \mu_t^\mathfrak{K} dt + \delta_t^\mathfrak{K}\cdot  dW_t$. Under Assumption  \ref{HypBarU}, $U$ is a dynamic  utility on the domain $\mathcal{D}_U =  \{ (\omega,t,z), \; z \geq \mathfrak{K}_t \}$, with  local characteristics
 \begin{numcases}{}
 \beta(t,z)  =     \bar \beta(t,z- \mathfrak{K}_t) -  U_z(t,z)\mu^{\mathfrak{K}}_t  - \frac{1}{2} U_{zz}(t,z)\|\delta^{\mathfrak{K}}_t\|^2 -  \gamma_z(t,z)\cdot \delta^{\mathfrak{K}}_t\label{barBeta}\\
 \gamma(t,z) =\bar \gamma(t,z-\mathfrak{K}_t)   - U_z(t,z)\delta_t^{\mathfrak{K}}, \label{barGamma} \quad \forall z \geq \mathfrak{K}_t. \label{liengammaet bar}
 \end{numcases}
\end{Proposition}

\section{Optimal PAYG pension policies}
\label{SectionMainResults}

Let us recall that a pensioner of age $a$ at time $t$ receives  pension $p(t,a)= p^{min}_t(a)\rho_t\geq p_t^{min}(a)$. The buffer fund in which the social planner can borrow/invest has the dynamics \eqref{EqDynK} given by: 
\begin{equation*}
d F^{\pi, \rho}_t=  F^{\pi, \rho}_t  r_t dt + (C_t-\rho_t P^{min}_t) dt +  \pi_t  \cdot  ( dW_t + \eta_t dt),
\end{equation*}
with $(C_t)$ the contribution process,  $(P_t)= (\rho_t P^{\min}_t)$ the  pension process introduced in \eqref{Pensionprocess}, and $(\pi_t)$ the investment strategy in the incomplete market.  The aggregated utility of pensioners  is given by   $V(t,\rho)= \int^\infty_{a_r}  v(t, a,  \rho) n(t,a) da$ (see \eqref{defaggregateV} and subsequent examples).\\
The first aim is  to  characterize the buffer fund utility $U$ of the social planner, such that the preference criteria $(U,V)$ is consistent, 
and  then to determine
 the admissible  strategy $(\pi,\rho) \in \A$   optimizing   the dynamic  criterion $U(t,F^{\pi,\rho}_t) + \int_0^t V(s,\rho_s) ds$.

\subsection{Consistency SPDE }
  Itô-Ventzel's formula allows us 
 to transform  the   supermartingale property implied by the consistency condition   into conditions
on the  differential characteristics of the utility process $U$. For standard deterministic  utility functions, the infinitesimal counterpart of the dynamic programming principle is  a nonlinear  Partial Differential Equation (PDE), called dynamic programming equation  or Hamilton-Jacobi-Bellman  (HJB) equation.  In the framework  of  dynamic utility, the  consistency characterization  is given  in terms of  an HJB  Stochastic Partial Differential Equation (SPDE), as detailed in \cite{zar-07} and \cite{MrNek1,MrNek07}.   The presence of pensions to be paid  impacts this SPDE in a non-linear way, the non-linear factor involving the utility    of pensioners $V$.
Note that the utility $U$ and $V$ are not of the same nature : the consistency is conveyed by $U$, which requires then stronger regularity conditions on $U$ than  on $V$. This HJB-SPDE provides a constraint on the drift $\beta$ of $U$,  as  explained  below.
\paragraph{Candidate to be the optimal strategy}
Applying It\^{o}-Ventzel's formula to the preference process $Z^{\pi,\rho}_{t}:=\int_0^t    V(s, \rho_s)   ds + U(t,F^{\pi,\rho}_t)$  with yields 
 \begin{eqnarray*}
\hspace{-3cm} dZ^{\pi,\rho}_{t}&=&\big(\beta(t,F^{\pi,\rho}_{t})+U_{z}(t,F^{\pi,\rho}_{t})(F^{\pi,\rho}_t   r_t+C_t)\big)dt \\
&+ &\big(\gamma(t,F^{\pi,\rho}_{t})+U_{z}(t,F^{\pi,\rho}_{t})\pi_t   \big) \cdot dW_{t}\\
&+&\big(\mathcal{P}(t,F^{\pi,\rho}_{t},\rho_{t})+\mathcal{Q}(t,F^{\pi,\rho}_{t},\pi_{t})\big)dt. 
\end{eqnarray*}\vspace{-8mm}
with \vspace{-2mm}
\begin{align}
\label{P}& \mathcal{P}(t,z,\rho):= V(t,\rho)-U_{z}(t,z)  P^{min}_t \rho\\
\label{Q} & \mathcal{Q}(t,z,\pi):=\demi U_{zz}(t,z)|| \pi_t||^{2}  +\pi_{t}\cdot (\gamma_{z}(t,z)+U_{z}(t,z) \eta_t). 
\end{align}
A natural  candidate for optimal policy $(\pi^*, \rho^*)$ are processes which  maximize  the drift of the preference process  $Z^{\pi,\rho}$.  Thus, 
$\pi^*$  should maximize $\mathcal{Q}(t,z,\cdot)$  and $\rho^*$ should maximize $\mathcal{P}(t,z,\cdot)$ on $[ 1, +\infty]$,  leading to 
\begin{equation}\label{candidate}
\pi^{*}_t(F_t^*)=-\frac{\gamma^{\mathcal R}_{z}(t,F_t^*)+U_{z}(t,F_t^*) \eta_t}{U_{zz}(t,F_t^*)} \quad \mbox { and } \quad 
\rho^{*}_t(F_t^*)=V_{\rho}^{-1}\big(t, P^{min}_t U_{z}(t,F_t^*)\big) \vee 1.
\end{equation}
Note that $\mathcal{Q}(t,z,\pi^*_t) = - \frac{1}{2}U_{zz}(t,z)\|\pi^*_t\|^2$.\\
With a slight abuse of notation, we use  interchangeably $\pi^*_t$ (resp. $\rho^*_t$)  or $\pi^*_t ( F^*_t)$  (resp. $ \rho^*_t ( F^*_t)$).\\
To alleviate the notations, we introduce the  optimal pension process   without minimum guarantee, denoted $\rho^f$ (free constraints). 
\begin{Definition}
The maximizer of    the operator $\mathcal{P}(t,z,\rho):=V(t, \rho)-U_{z}(t,z) P^{min}_t  \rho$ for $\rho \in \mathbb R^+$  is denoted 
\begin{equation}
\label{pf}
\rho^f_t (z) = V_{\rho}^{-1}\big(t, P^{min}_t U_{z}(t,z)\big).
\end{equation}
Remark that $\mathcal{P}(t,z,\rho^f_t (z)) = \tilde V (t, P^{min}_tU_{z}(t,z))$.
\end{Definition}
\begin{Remark}
\label{RemInadaV}
Note  that if $U$ satisfies the Inada condition at  $\mathfrak{K}$, then  $v$ also satisfy the Inada condition at  $\underline \rho$,in order to ensure that  the optimal pension is   well defined (namely the quantity $V_{\rho}^{-1}\big(t, P^{min}_t U_{z}(t,z)\big) $). 
Remark also that Inada conditions for $U$ and $v$ at  $+\infty$ can be relaxed into the following condition $\lim\limits_{ \rho \rightarrow +\infty} V_\rho(t,\rho) \leq P^{min}_t \lim\limits_{z \rightarrow +\infty} U_z(t,z)$.
\end{Remark}
\paragraph{Consistency condition on the drift $\beta$ of $U$} 
If the candidate $(\pi^*,\rho^*)$ is  indeed the optimal strategy, then to satisfy the time consistency, the drift of the preference process  $Z^{\pi,\rho}$ should be nonpositive for any admissible strategy $(\pi, \rho)$ and equal to zero for the optimal stragegy $(\pi^*, \rho^*)$. This leads to the  following sufficient condition on  the drift of $U$ \vspace{-2mm}
\begin{eqnarray}
\hspace{-0.8cm} \beta(t,z) &= & -U_{z}(t,z)   (zr_t+C_t)
 -  \mathcal{Q}(t,z,\pi^*)-\mathcal{P}(t,z,\rho^{*}_{t}(z))\nonumber\\
&=  &  -U_{z}(t,z) \big(zr_t+C_t- P^{min}_t  \rho^*\big)
+\demi U_{zz}(t,z)\| \pi^*_t (z)\|^{2}-V(t, \rho^*_{t}(z))\label{HJB-Beta}.
\end{eqnarray}

\subsection{Main results}

We gather  here the main results the paper, proofs and examples being postponed in the next subsections. The  first below result shows that under the consistency HJB condition \eqref{HJB-Beta},   the bound $\mathfrak{K}$ shifting the utility $U$  is necessarily a buffer fund. This is an interesting new result.

 \begin{Theorem}\label{SusCond}
Let $U$ be the buffer fund utility introduced in \ref{defU} and verifying Assumption \ref{HypBarU}, and $V$  the aggregated pensioners' utility, verifying Assumption \ref{hypVbis}. Assume that  the drift $\beta$ of $U$ satisfies the HJB constraint
$$ \beta(t,z) =  -U_{z}(t,z) \big(zr_t+C_t- P^{min}_t \rho^*_t (z)\big)
+\demi U_{zz}(t,z)\| \pi^*_t (z)\|^{2}-V(t, \rho^*_t (z)). \quad \quad \eqref{HJB-Beta} $$
 Then the sustainability bound $\mathfrak{K}$ is necessarily   a  buffer fund  receiving the contribution $C$ and paying the minimal pension amount  $ P^{min}$, that is:
 \begin{equation}\label{BFdynamic}
 d\mathfrak{K}_t =  (\mathfrak{K}_t r_t + C_t -   P^{min}_t) dt + \delta_t^\mathfrak{K}  \cdot (dW_t + \eta_t dt),~\delta^\mathfrak{K}\in \sigR.
 \end{equation}
  \end{Theorem}
  
 \noindent The proof is postponed to  Section \ref{sec:proofTh3.2}.

\begin{Remark}\label{Rk:sureplication} {\rm
Theorem \ref{SusCond} states that if the  utility system   $(U,V)$ is consistent, then the shift $\mathfrak{K}_t$ of the utility $U$ is  necessarily a  buffer fund with minimal pensions.   Nevertheless, in the original problem formulation and Definition \ref{defU} of $U$,  $\mathfrak{K}_t$ is not assumed to be  a buffer fund itself: for instance if $\mathfrak{K}$ is an index following the GDP, there is no reason why it will  follow dynamics \eqref{BFdynamic}. Therefore, to satisfy the consistency constraint, the dynamic utility should be shifted not with  the sustainability constraint $\mathfrak{K}_t$, itself but  with  a buffer fund process $(X_t)$ that super-replicates (in a pathwise way) $(\mathfrak{K}_t)$ that is :  $X_t \geq \mathfrak{K}_t$,   for all $t$, $\mathbb P$ a.s. 
 This means that the sustainability constraint is transformed into a stronger one: $F_t \geq X_t  \;(\geq \mathfrak{K}_t) $,   for all $t$, $\mathbb P$ a.s..  
The problem is equivalent to searching for a self-financing portfolio (without contributions and pensions)  $X'$   super-replicating (pathwise)  the process $B_t := \mathfrak{K}_t + \int_0^t (P_s^{min}- C_s) ds $.
For example,  it may be  relevant to choose the ``minimal"  super-replicating self-financing portfolio $(X'_t)$ (if it exists) in the following sense:
{\footnotesize $$X^*_0:=\inf \{X_0 \geq \mathfrak{K}_0 \,s.t. \; \exists \pi' \in \sigR \ \mbox{ satisfying }  X'_t:=X_0+ \int_0^t r_{s}X_{s}' ds + \pi_s' \cdot (dW_{s}+\eta_{s}ds)  \geq  B_t , 
 \; dt \otimes d\mathbb{P} a.s.  \}$$}
\noindent The existence of a  super-replicating self-financing portfolio is not guaranteed, especially in our context of  incomplete market, in which demographic risk can not be  completely hedged by financial assets.  Applying Theorem 5.12 of Karatzas and Kou \cite{karatzas1998hedging},  a  sufficient existence condition on $[0,T]$ is
 \begin{equation}\label{condSR}
 \sup_{\nu \in \mathcal{R}^\perp} \sup_{\tau \in \Tau_{[0,T]}} \mathbb E \left(           Y_\tau^\nu  \big( \mathfrak{K}_\tau + \int_0^\tau (P_s^{min}- C_s) ds  \big)^+  \right)  < \infty, \quad \forall\,  T \geq 0,
  \end{equation}
 where $  \Tau_{[0,T]}$ is the class of stopping times $\tau$ with values in the interval  $[0, T ]$ and $Y^\nu$ the state price density process \eqref{Ydyn}. Note that this supremum corresponds to $X^*_0$ which is  the super-replicating price of $(B_t)$. \\
In the backward framework, El Karoui and Jeanblanc \cite{el1998optimization},  or He and Pagès \cite{he1993labor} consider a complete market, which ensures the existence of a super-replicating portfolio.  In incomplet market, an analoguous assumption as \eqref{condSR}  is needed  in Mostovyi and Sirbu \cite{mostovyi2020optimal}  (Assumption 2.5), to  ensures the existence of a super-replicating portfolio (see Lemma 3.2 \cite{mostovyi2020optimal} ). 
}
\end{Remark}
\noindent  The following theorem shows that the policy  given by \eqref{candidate} is  indeed the optimal strategy, and that the corresponding buffer fund satisfies the sustainability constraint  \eqref{ConstraintK2}. 
 \begin{Theorem}\label{ThOpt}
Let $U$ be the buffer fund utility verifying Assumptions \ref{HypBarU},  and $V$ be the aggregated pensioners' utility, verifying Assumption \ref{hypVbis}. Assume $F_0 > \mathfrak{K}_0$ and  that the drift $\beta$ of $U$ satisfies the HJB constraint \eqref{HJB-Beta}:
\begin{equation*}
 \beta(t,z) =  -U_{z}(t,z) \big(zr_t+C_t- P^{min}_t \rho_t^*(z)\big)
+\demi U_{zz}(t,z)\| \pi^*_t (z)\|^{2} - V(t, \rho_t^*(z)),
\end{equation*}
\begin{numcases}{ with }
\pi^{*}_t (z)=-\frac{\gamma^{\mathcal R}_{z}(t,z)+U_{z}(t,z) \eta_t}{U_{zz}(t,z)},\label{pi*}\\
\rho^{*}_t (z)= V_{\rho}^{-1}\big(t, P^{min}_t U_{z}(t,z)\big) \vee 1. \label{p*}
\end{numcases}
Then the portfolio/pension plan $(\pi^*,\rho^*)$ is the optimal strategy. In particular, 
 the buffer fund $F^{*}$ following the investment strategy $\pi^{*}_t =\pi^*_t (F^*_t)$ and paying the pension  amount $\rho^{*}_t =\rho^*_t (F^*_t)$ is strictly greater than $\mathfrak{K}_{t}$ and satisfies the dynamics
 \begin{numcases}{}
dF^{*}_{t}=\mu^{*}_t (F^{*}_{t})dt+\pi^{*}_t (F^{*}_{t}).dW_{t},\label{dynF*} \\
\mu^{*}_t (z):=zr_{t}+C_{t}- P^{min}_t \rho^{*}_t (z)+\pi^{*}_t (z)\eta_{t}.\nonumber
\end{numcases}
\end{Theorem}
\noindent The proof is postponed to  Section \ref{sec:proofTh3.3}.\\
Since $V_\rho^{-1}(t, \cdot)$ and $U_z(t,\cdot)$ are decreasing functionals, the optimal pension  ${p_t^*(a)  = \rho_t^* p_t^{min}(a)}$ is increasing in the fund's wealth, which is natural. In particular, if the number of workers and/or the wage amount decreases, the contribution $C_t$ can become lower to the pension amount $P_t^*$, which in turn can lead to a decrease in the optimal pension's value. On the other hand, when minimum pension amount $P^{min}_t$ to be paid increases (for instance due to a decrease of death rates at older ages), $p_t^*(a)$ decreases.\\
The optimal investment strategy depends on  the market risk premium vector $(\eta_t)$ which can be correlated to wages and demographic rates. The  optimal investment strategy also depends on the inverse risk aversion $- \frac{U_z(t,F^*_t)}{U_{zz}(t,F^*_t)}$ and on the (buffer fund) utility diffusion coefficient through the term $\gamma^{\cal R}_z(t,F^*_t)$. Besides,   $\gamma^{\cal R}_z(t,F^*_t)$ also depends on  the financial, economic and demographic parameters, due to the social planner's preferences and the HJB consistency constraint  \eqref{HJB-Beta}. We refer to Section \ref{sec:exCRRA} for details on a specific example. 
\begin{Corollary}
Under the assumptions and notations of Theorem  \ref{ThOpt}, 
\begin{equation}
\label{eqlimpmin}
\lim_{z \rightarrow \mathfrak{K}_t}  \pi^*_t (z)  = \delta^{\mathfrak{K}}_t, \quad \lim_{z \rightarrow \mathfrak{K}_t} \rho^*_t (z) = 1, \quad \forall t  \geq 0 \; \mathbb{P}-\text{a.s.}
\end{equation}
This means that the optimal strategy converges  to the ``minimal" buffer fund strategy when $F^*$ tends to the sustainability bound. 
\end{Corollary}
\begin{proof} First, using the expression of the optimal portfolio  \eqref{pi*},
\begin{eqnarray*}
\| \pi^*_t (z) - \delta_t^{\mathfrak{K}}\| & = & ||\frac{\gamma^{\mathcal R}_{z}(t,z)+U_{z}(t,z) \eta_t}{U_{zz}(t,z)}+\delta^\mathfrak{K}_{t}||\\
&\le& ||\frac{\gamma^{\mathcal R}_{z}(t,z)+U_{zz}(t,z) \delta^\mathfrak{K}_{t}}{U_{zz}(t,z)}||+||\eta_{t}||\frac{U_{z}(t,z)}{|U_{zz}(t,z)|}\\
&\stackrel{\eqref{boundenzero}}{\le}& (\zeta_{t}+||\eta_{t}||)\frac{U_{z}(t,z)}{|U_{zz}(t,z)|}\stackrel{z\to\mathfrak{K}_{t }}{\longrightarrow}0
\end{eqnarray*}
since $\lim_{z \rightarrow \mathfrak{K}_t} \frac{{U}_{zz}(t,z)}{ U_z(t,z)  }=-\infty$ from Remark \ref{RemarkInada}.  
Finally,  by the Inada condition on $U$ and $V$,
$$\lim_{z \rightarrow \mathfrak{K}_t}  \rho^*_t (z) \vee 1 = \lim_{z \rightarrow \mathfrak{K}_t}  V_{\rho}^{-1}\big(t,P^{min}_t U_{z}(t,z)\big)\vee 1 = \lim_{\rho \rightarrow  +\infty} V_{\rho}^{-1}(t, \rho) \vee1 = \underline \rho_t \vee  1= 1.$$
\end{proof}
\paragraph{Example \ref{Ex1}} We come back to  the  first example, in which the pension and individual pensioners' utility $\bar v$ do not depend on the age of the pensioner. The social planner attributes the global  weight $\omega^r_t = \alpha_c   \int_{a_r}^{\infty}  c_t (a) n(t,a) da $   to  pensioners living at time $t$, based on their past contributions (see \eqref{eqpoidsexample1}). In this example, the  aggregate utility of pension is given by \eqref{eqVexample1}:
$$V(t,\rho) =\int^\infty_{a_r} \bar v(t,\rho p^{min}_t )\alpha_c c_t (a)n(t,a) da=  \omega^r_t \bar v(t, \rho p^{min}_t).$$
We have $ p^{min}_t\ V_{\rho}^{-1}(t,z)= \bar v_{\rho}^{-1}\big(t, \frac{z}{p^{min}_t \omega_t^r} \big).$  Besides, $P^{min}_t= N_t^r p^{min}_t$ where  $N_t^r=\int_{a_r}^{\infty}  n(t,a) da $ is the number of pensioners at time $t$. This implies that the optimal pension for each pensioner is
\begin{equation}\label{pstarex1}
p^*_t (z)= p^{min}_t\rho^{*}_t (z)= \bar v_{\rho}^{-1}\big(t, \frac{N_t^r}{\omega_t^r} U_{z}(t,z)\big) \vee p_t^{min}.
\end{equation}
The pension amount then increases with the quantity $\frac{\omega_t^r}{N_t^r}$,  corresponding to the average individual contribution of a  pensioner living at  time $t$.
 
\noindent {\bf Example} \ref{Ex2}  In the second  example of age-dependent  pension and  utility $\bar v$,   the aggregate utility of pension $V$  is a complex aggregation between cohorts given by \eqref{EqVexample2}:
$$V(t,\rho) = \int^\infty_{a_r} \bar v(t, a, \rho \, p_{ret}(a_r + t-a)e^{\int_{a_r+t-a}^t\lambda_u du} ) \, n(t,a) da.$$
In all generality, there is no straightforward formula for    $V_{\rho}^{-1}$ in terms of  the $\bar v_{\rho}^{-1}$, except in particular cases of dynamic utilities such as dynamic CRRA utilities (see Section \ref{sec:exCRRA}).

\subsection{Proof of  Theorem \ref{SusCond}}\label{sec:proofTh3.2}
Theorem \ref{SusCond} states that  in order to satisfy the consistency condition  \eqref{HJB-Beta},  the sustainability bound $\mathfrak{K}$ is necessarily   a  buffer fund  receiving the contribution $C$ and paying the minimal pension amount  $ P^{min}$.  We recall the dynamics of  $\mathfrak{K}$ is given by \eqref{EqDynK}:
$d\mathfrak{K}_t = \mu_t^{\mathfrak{K}} dt + \delta^{\mathfrak{K}}_t \cdot dW_t.$
As in Proposition \ref{eq:Lipschitzcarac}, we consider  hereafter the stochastic utility  $\bar U(t,z)= U(t,z+\mathfrak{K}_t)$ whose local characteristics are given,   by 
\begin{numcases}{}
\beta(t,z) =  {\bar \beta}(t, z-\mathfrak{K}_t) - {\bar U}_z(t, z-\mathfrak{K}_t) \mu^\mathfrak{K}_t + \frac{1}{2}{\bar U}_{zz}(t,z-\mathfrak{K}_t)||\delta^\mathfrak{K}_t||^2 - {\bar \gamma}_z(t,z-\mathfrak{K}_t)\delta^\mathfrak{K}_t,\nonumber\\
\gamma(t,z) = \bar{\gamma}(t,z-\mathfrak{K}_t) - \bar{U}_z(t,z-\mathfrak{K}_t)\delta^{\mathfrak{K}}_t,\nonumber
\end{numcases}
by the Itô-Ventzel's formula. This combined with the HJB-constraint below,
\begin{equation*}
\label{Consistency}
\beta(t,z) =-U_{z}(t,z) (zr_t+C_t)
+\demi U_{zz}(t,z)\| \frac{\gamma^{\mathcal R}_{z}(t,z)+U_{z}(t,z) \eta_t}{U_{zz}(t,z)}\|^{2} - V(t,\rho^*_t(z)) + U_z(t,z) P^{min}_t\rho^*_t(z),
\end{equation*}
\noindent yields that
\begin{align}
& \bar \beta(t, z-\mathfrak{K}_t) - \bar U_z(t, z-\mathfrak{K}_t) \mu^\mathfrak{K}_t + \frac{1}{2}\bar U_{zz}(t,z-\mathfrak{K}_t)||\delta^\mathfrak{K}_t||^2 - \bar \gamma_z(t,z-\mathfrak{K}_t)\cdot \delta^\mathfrak{K}_t   \nonumber \\
& \nonumber \hspace{1cm}= - \bar U_{z}(t,z-\mathfrak{K}_t) (zr_t+C_t) - V(t,\rho^*_t(z)) + \bar U_z(t,z-\mathfrak{K}_t) P^{min}_t \rho^*_t(z) \\
& \nonumber \hspace{1cm} + \demi \bar U_{zz}(t,z-\mathfrak{K}_t)||\frac{\bar{\gamma}^{\mathcal R}_z(t,z-\mathfrak{K}_t) - \bar{U}_{zz}(t,z-\mathfrak{K}_t)\delta^{\mathfrak{K},\mathcal R}_t + \bar U_z(t,z-\mathfrak{K}_t)\eta_t}{\bar{U}_{zz}(t,z-\mathfrak{K}_t)}||^2.
 \end{align}
 Reorganizing the terms, we have  (omitting the variables $(t,z-\mathfrak{K}_t)$ to simplify  notations): 
\begin{equation*}
\label{Consistency2}\bar \beta +  V(t,\rho^*_t(z))  = \bar U_z \; (\mu^\mathfrak{K}_t - zr_t  - C_t + P^{min}_t \rho_t^*(z))  + \demi \frac{1}{\bar U_{zz}}\big(||\bar{\gamma}_z^{\mathcal R} - \bar{U}_{zz} \delta^{\mathfrak{K},\mathcal R}_t + \bar U_z\eta_t||^2 - ||\bar U_{zz}\delta^\mathfrak{K}_t||^2\big)
 + \bar \gamma_z \cdot \delta^\mathfrak{K}_t.
 \end{equation*}
Rewriting the last term as follows,
\begin{align*}
 & \demi \frac{1}{\bar U_{zz}}\big(||\bar{\gamma}_z^{\mathcal R}- \bar{U}_{zz} \delta^{\mathfrak{K},\mathcal R}_t + \bar U_z\eta_t||^2 - ||\bar U_{zz}\delta^\mathfrak{K}_t||^2 \big) + \bar \gamma_z \cdot \delta^\mathfrak{K}_t   \\
&  \hspace{1cm} = \demi \frac{1}{\bar U_{zz}}\big(|| \bar{\gamma}_z^{\mathcal R} +  \bar U_z\eta_t||^2  - 2(\bar{\gamma}_z^{\mathcal R} +  \bar U_z\eta_t)\cdot \bar{U}_{zz} \delta^{\mathfrak{K}}_t -  ||\bar U_{zz}\delta^{\mathfrak{K},\perp}_t||^2 \big)+ \bar \gamma_z \cdot \delta^\mathfrak{K}_t   \\
&  \hspace{1cm}=  \demi \frac{1}{\bar U_{zz}}\big(|| \bar{\gamma}_z^{\mathcal R} +  \bar U_z\eta_t||^2  - ||\bar U_{zz}\delta^{\mathfrak{K},\perp}_t||^2 \big)   + \bar \gamma_z^{\perp}\cdot \delta^\mathfrak{K}_t  - \bar U_z  \eta_t \cdot \delta^{\mathfrak{K}}_t,
\end{align*}
we get that the consistency condition \eqref{HJB-Beta} is equivalent to
\begin{eqnarray*}
& &\bar \beta (t,z-\mathfrak{K}_t) +  V(t,\rho^*_t(z))  \\
&  = &\bar U_z \; (\mu^\mathfrak{K}_t - zr_t  - C_t + P^{min}_t \rho ^*(z) - \eta_t \cdot \delta^{\mathfrak{K}}_t) + \demi \frac{1}{\bar U_{zz}}\big(|| \bar{\gamma}_z^{\mathcal R} +  \bar U_z\eta_t||^2 -  ||\bar U_{zz}\delta^{\mathfrak{K},\perp}_t||^2 \big)    + \bar \gamma_z^{\perp}\cdot \delta^\mathfrak{K}_t  \\
& =& \bar U_z \Big[\mu^\mathfrak{K}_t - zr_t  - C_t + P^{min}_t \rho_t^*(z) - \eta_t \cdot \delta^{\mathfrak{K}}_t +  \demi \frac{1}{\bar U_{zz}\bar U_z }\big(|| \bar{\gamma}_z^{\mathcal R} +  \bar U_z\eta_t||^2 - ||\bar U_{zz}\delta^{\mathfrak{K},\perp}_t||^2 \big)   + \frac{ \bar \gamma_z^{\perp}}{\bar U_z}\cdot \delta^\mathfrak{K}_t  \Big].
\end{eqnarray*}
Let us analyze the  behavior of the  left hand size and the right hand size of this identity, when $z \to \mathfrak{K_t}$.
We have  by \ref{defU} and \ref{hypV} that $ \lim\limits_{z \to \mathfrak{K_t}} \rho^*_t(z) =  \lim\limits_{z \to \mathfrak{K_t}} V_{p}^{-1}\big(t, P^{min}_t U_z(t,z) \big)\vee  1  = \underline{\rho}_t \vee 1= 1$.\\
Furthermore, $\bar U(t,0)$ is well defined. Hence,  by continuity the limit of $ \bar \beta(t,z-\mathfrak{K}_t)$ when $z -\mathfrak{K_t}\to 0$ exists and is equal to $\bar \beta(t,0)$. 
Thus,  the left-hand side of the previous equation tends to a constant  $ l_t <\infty$. \\
For  the right-hand side,  since by the Inada condition $\bar U_z(t,z-\mathfrak{K}_t) \to \infty$ when $ z\to \mathfrak{K}_t$, the bracketed term can therefore only tend towards zero, that is 
\begin{equation*}
\lim_{z \to \mathfrak{K}_t} \Big[\mu^\mathfrak{K}_t - zr_t  - C_t + P^{min}_t \rho_t^*(z) - \eta_t \cdot \delta^{\mathfrak{K}}_t +  \demi \frac{1}{\bar U_{zz}\bar U_z }\big(|| \bar{\gamma}_z^{\mathcal R} +  \bar U_z\eta_t||^2 - ||\bar U_{zz}\delta^{\mathfrak{K},\perp}_t||^2 \big)   + \frac{ \bar \gamma_z^{\perp}}{\bar U_z}\cdot \delta^\mathfrak{K}_t\Big]   =  0,
\end{equation*}
which we rewrite
\begin{equation*}
\lim_{z \to \mathfrak{K}_t} \Big[\mu^\mathfrak{K}_t - zr_t  - C_t +  P^{min}_t \rho_t^*(z) - \eta_t \cdot \delta^{\mathfrak{K}}_t + \frac{ \bar \gamma_z^{\perp}}{\bar U_z}\cdot \delta^\mathfrak{K}_t+ \demi \big( \frac{|| \bar{\gamma}_z^{\mathcal R} \|^2 }{\bar U_{zz}\bar U_z } + 2 \frac{\bar \gamma_z^{\mathcal{R}}\cdot \eta_t}{\bar U_{zz}} + \frac{\bar U_{z}}{\bar U_{zz}}\| \eta_t\|^2 - \frac{\bar U_{zz}}{\bar U_z}\| \delta_t^{\mathfrak{K},\perp}\|^2\big) 
   \Big]   =  0,
\end{equation*}
or equivalently,
\begin{equation}\label{lalimite}
\lim_{z \to \mathfrak{K}_t} \Big[A_{t}(z)+ \demi \big( B_{t}(z) - \frac{\bar U_{zz}}{\bar U_z}\| \delta_t^{\mathfrak{K},\perp}\|^2\big) 
   \Big]   =  0,
\end{equation}
where we have used the notations
\begin{numcases}{}
A_{t}(z)=\mu^\mathfrak{K}_t - zr_t  - C_t + P^{min}_t \rho_t^*(z) - \eta_t \cdot \delta^{\mathfrak{K}}_t + \frac{ \bar \gamma_z^{\perp}}{\bar U_z}\cdot \delta^\mathfrak{K}_t,\nonumber\\
B_{t}(z)=\frac{|| \bar{\gamma}_z^{\mathcal R} \|^2 }{\bar U_{zz}\bar U_z } + 2 \frac{\bar \gamma_z^{\mathcal{R}}\cdot \eta_t}{\bar U_{zz}} + \frac{\bar U_{z}}{\bar U_{zz}}\| \eta_t\|^2.\nonumber
\end{numcases}
We shall now study the limits of $A_t(z)$ and $B_t(z)$ when z tends to $\mathfrak{K}_t$. For this, we recall according to Remark \ref{RemarkInada}, that if the  Inada condition holds, then
$$\lim\limits_{z \to \mathfrak{K}_t} \frac{\bar U_{zz}(t,z-\mathfrak{K}_t)}{\bar U_z(t,z-\mathfrak{K}_t)}  = \lim\limits_{z \to \mathfrak{K}_t } \frac{U_{zz}(t,z)}{U_z(t,z)} = +\infty.$$
Also, %
under   Assumption \ref{HypBarU}, there exist a random  bound  $\zeta_t$ satisfying a.s.  {  $\int_{0}^T\zeta_{t}^2dt<\infty$} for any $T$ such that
$$ \lim_{z \to \mathfrak{K}_t}  \frac{|| \bar \gamma_{z}(t,z-\mathfrak{K}_t) \|}{\bar  U_{z}(t,z-\mathfrak{K}_t)} \leq \zeta_t $$
Thus, $|B_{t}(z)|=|\frac{\|\bar{\gamma}_z^{\mathcal R} \|^2 }{\bar U_{zz}\bar U_z }  + 2 \frac{\bar \gamma_z^{\mathcal{R}}\cdot \eta_t}{\bar U_{zz}}  + \frac{\bar U_{z}}{\bar U_{zz}}\| \eta_t\|^2 | \leq \frac{\bar U_z}{\bar U_{zz}}(  \zeta_t   +  3\| \eta_t\|^2 ) \to 0 $ when $(z-\mathfrak{K}_t ) \to 0$. \\
Furthermore,  
$$ |A_{t}(z)|=|(\mu^\mathfrak{K}_t - zr_t  - C_t + P^{min}_t \rho_t^*(z) - \eta_t \cdot \delta^{\mathfrak{K}}_t) + \frac{ \bar \gamma_z^{\perp}}{\bar U_z}\cdot \delta^\mathfrak{K}_t | \leq | (\mu^\mathfrak{K}_t - zr_t  - C_t + P^{min}_t \rho_t^*(z) - \eta_t \cdot \delta^{\mathfrak{K}}_t)|  + \zeta_t \| \delta^{\mathfrak{K},\perp}_t \|, $$ 
consequently $A_{t}(z)$  has a finite  limit when $(z -  \mathfrak{K}_t ) \to 0 $. This combined with $|B_{t}(z)| \to 0 $ when $(z-\mathfrak{K}_t ) \to 0$,   \eqref{lalimite} and the fact that $\lim\limits_{z \to \mathfrak{K}_t} \frac{\bar U_{zz}(t,z-\mathfrak{K}_t)}{\bar U_z(t,z-\mathfrak{K}_t)} = +\infty$, implies that  necessarily 
\begin{equation}
  \delta_t^{{\mathfrak K},\perp}  =  0.
\end{equation}
It  follows  from  \eqref{lalimite} that 
\begin{equation}
 \lim_{z \to \mathfrak{K}_t} (\mu^\mathfrak{K}_t - zr_t  - C_t + P^{min}_t \rho_t^*(z) - \eta_t \cdot \delta^{\mathfrak{K}}_t  )=\mu^\mathfrak{K}_t - \mathfrak{K}_tr_t  - C_t + P^{min}_t - \eta_t \cdot \delta^{\mathfrak{K}}_t= 0,
\end{equation}
where we have used, according to \eqref{eqlimpmin},
$
\lim_{z\to \mathfrak{K}_t} \rho^*_t(z) = 1.
$
This concludes the proof of   Theorem \ref{SusCond}.
\subsection{Proof of Theorem \ref{ThOpt}}\label{sec:proofTh3.3}
In order to show that the  sustainability condition  $(F^*_{t}\ge \mathfrak{K}_{t},~ dt \otimes d\P ~\text{a.s.}$) is satisfied under the assumptions of Theorem \ref{ThOpt}, we study the intermediate process $(U_z(t,F^*_t))$. This process is actually the optimal   state price density process  $(Y^*_t)$. We refer to \cite{MrNek07} for more details.\\
Combining  the bounds  \eqref{boundenzero}   on the derivatives of the diffusion coefficient $\bar \gamma(t,z)$  of $\bar U$ and  Theorem \ref{SusCond}, we first show that $U_z(t,F^*_t)$ is the unique strong (non-explosive) solution of an SDE with Lipschitz coefficients.  This property, combined with Inada conditions, ensures that the sustainability condition is necessarily satisfied.  This preliminary result is  stated in the following Proposition \ref{MargU}.
\begin{Proposition}\label{MargU} 
Under the assumptions and notations  of Theorem \ref{ThOpt}, 
\begin{itemize}
\item[(i)] The dynamics of the marginal utility $U_{z}$ is given by
\begin{eqnarray*}
dU_{z}(t,z)&=&\Big(-U_{zz}(t,z) \mu^{*}_{t}(z)-U_{z}(t,z) r_t
-\demi U_{zzz}(t,z)\| \pi^*_{t}(z)\|^{2}-\gamma_{zz}^{\mathcal R}\cdot\pi^*_{t}(z)\Big)dt\nonumber\\
&+&\gamma_{z}(t,z) \cdot dW_{t}
\end{eqnarray*}
\item[(ii)] Let $\tau^F = \inf\{ t\geq 0 , F_t^{*} = \mathfrak{K_t}\}$.  
 $U_{z}(t,F^*_{t})$ is the  solution of the following SDE on $[0,\tau^F[$ of
\begin{eqnarray}\label{eqU_{z}}
dU_{z}(t,F^*_{t})&=&-U_{z}(t,F^*_{t}) r_tdt+
\big(\gamma^{\perp}_{z}(t,F^*_{t})-U_{z}(t,F^*_{t})\eta_{t}\big) \cdot dW_{t}.
\end{eqnarray}
\end{itemize} 
\end{Proposition}
\noindent The proof of Theorem \ref{ThOpt} is derived from Theorem  \ref{SusCond}  and  Proposition \ref{MargU}, by following the same steps  than in the proof of Theorem 4.8 in \cite{MrNek07}.
\begin{proof}[\underline{Proof of Theorem \ref{ThOpt}}]
Inspired by \eqref{eqU_{z}},  let us consider the following SDE
$$dY_t=-Y_t r_tdt+\big( \gamma^{\perp}_{z}(t,  U^{-1}_{z}(t,Y_{t})  )-Y_t \eta_{t}\big) \cdot dW_{t}.$$
By Theorem \ref{SusCond}, $\delta^{\mathfrak{K},\perp}=0$ a.s..  Hence, using the relation \eqref{liengammaet bar} between $\gamma$ and $\bar \gamma$,  the condition  \eqref{boundenzero} can be rewritten as  $|| \gamma^\perp_z(t,z)|| \le \zeta_{t}U_{z}(t,z)$ for some non negative  process  $\zeta$ such that a.s. $\int_{0}^T\zeta^2_{t}dt<\infty$ for any $T$. 
This implies that the coefficients of this SDE are Lipschitz, yielding that this SDE admits  a unique strong (non explosive) solution.
Furthermore, by Proposition \ref{MargU},  $U_{z}(t,F^*_{t})$ is solution of this SDE, which yields that $\tau^F=+\infty$ i.e.   $F^*_t > \mathfrak{K}_t$, $\forall t>0$, $\mathbb P$ a.s.. Besides, 
as a consequence of   Assumptions  \ref{hypVbis} and \ref{HypBarU}, $\big(\pi^*_{t}(F^*_t), \rho^*_{t}(F^*)\big)$ verify the required integrability condition, which concludes the proof.  
\end{proof}
\noindent The last step consists in proving Proposition \ref{MargU}.
\begin{proof}[\underline{ Proof of Proposition \ref{MargU}}] \mbox{}\\
$(i)$ Since by assumption $U$ satisfies the HJB constraint \eqref{HJB-Beta},  its dynamics is 
\begin{eqnarray}
dU(t,z)&=&\Big(-U_{z}(t,z) \big(zr_t+C_t- P^{min}_t  \rho^*_{t}(z)\big)
+\demi U_{zz}(t,z)\| \pi^*_{t}(z)\|^{2}- V(t, \rho^*_{t}(z))\Big)dt\nonumber\\
&+&\gamma(t,z) \cdot dW_{t}.\label{dynU}
\end{eqnarray}
In addition, since  $U$ is of class ${\cal C}^{3,\delta}$ and by Assumption \ref{HypBarU}, this implies by  Corollary 1.3 in \cite{MrNek1}, that $ \beta_z$ and $ \gamma_z$ are the local characteristics of the space derivative $U_z$. \\
On  $\{ 1< \rho^f_t(z)\}$,  $\rho^*_t(z) =\rho^f_t(z) \in \mathcal{C}^1$ with     $U_{z}(t,z)P^{min}_t =  V_{\rho} (t, \rho^f_{t}(z))$ by \eqref{pf}. Thus $U_{z}(t,z)P^{min}_t  \partial_{z}\rho^*_{t}(z)=V_{\rho} (t, \rho^*_{t}(z))\partial_{z}\rho^*_{t}(z)$ on $\{ 1< \rho^f_t(z)\}$.\\
This last inequality also holds on  $\{ \rho^f_t(z) \leq 1\}$ since  on this set, $\rho^*_{t}(z) \equiv 1$ implying  $\partial_{z}\rho^*_{t}(z)=0$. 
Using those simplifications, the derivative with respect to $z$ of the dynamics \eqref{dynU} of $U$ yields by \cite[Theorem 2.2]{MrNek1},
\begin{eqnarray}\label{eqU_{z}temp}
dU_{z}(t,z)&=&\Big(-U_{zz}(t,z) \big(zr_t+C_t-P^{min}_t \rho^*_{t}(z)\big)-U_{z}(t,z) r_t\\
\nonumber &+&\demi U_{zzz}(t,z)\| \pi^*_{t}(z)\|^{2}+ U_{zz}(t,z)\partial_{z} \pi^*_{t}(z). \pi^*_{t}(z)\Big)dt 
+ \gamma_{z}(t,z).dW_{t}.
\end{eqnarray}
\noindent
Now,  using the definition  \eqref{pi*} of $\pi^*$,
\begin{equation}
\label{eqU_{z}2}
\partial_{z} \pi^*_{t}(z)=-\frac{U_{zzz}(t,z)}{U_{zz}(t,z)}\pi^*_{t}(z)-\frac{\gamma^{\mathcal R}_{zz}(t,F_t^*)+U_{zz}(t,F_t^*) \eta_t}{U_{zz}(t,F_t^*)},
\end{equation}
which implies that
\begin{eqnarray}\label{eqU_{z}1}
&&\demi U_{zzz
}(t,z)\| \pi^*_{t}(z)\|^{2}+U_{zz}(t,z)\partial_{z} \pi^*_{t}(z). \pi^*_{t}(z)\nonumber\\
&&=-\demi U_{zzz}(t,z)\| \pi^*_{t}(z)\|^{2}-\gamma_{zz}^R.\pi^*_{t}(z)-U_{zz}(t,z)\pi^*_{t}(z).\eta_{t}.
\end{eqnarray}
Using the notation $\mu^{*}_{t}(z)=zr_{t}+C_{t}-P^{min}_t \rho^{*}_{t}(z)+\pi^{*}_{t}(z)\eta_{t}$  for the drift of $F^*$ and injecting \eqref{eqU_{z}2} and \eqref{eqU_{z}1} in the dynamics \eqref{eqU_{z}temp}, yields
\begin{eqnarray*}
dU_{z}(t,z)&=&\Big(-U_{zz}(t,z) \mu^{*}_{t}(z)-U_{z}(t,z) r_t
-\demi U_{zzz}(t,z)\| \pi^*_{t}(z)\|^{2}-\gamma_{zz}^{\sigR}.\pi^*_{t}(z)\Big)dt\nonumber\\
&+&\gamma_{z}(t,z).dW_{t}.
\end{eqnarray*}
This establishes the first statement of the theorem. \\
(ii)  Applying  Itô-Ventzel's formula to $U_{z}(t,F^*_{t})$  on $\{t <\tau^F\}$ gives (recalling the dynamics  $dF^{*}_{t}=\mu^{*}_t (F^{*}_{t})dt+\pi^{*}_t (F^{*}_{t}).dW_{t}$):
\begin{eqnarray*}
dU_{z}(t,F^*_{t})&=&-U_{z}(t,F^*_{t}) r_tdt+
\big(\gamma_{z}(t,F^*_{t})+U_{zz}(t,F^*_{t})\pi^*_{t}(F^*_{t})\big). dW_{t}. 
\end{eqnarray*}
Using the identity $\gamma_{z}(t,F^*_{t})+U_{zz}(t,F^*_{t})\pi^*_{t}(F^*_{t})=\gamma^{\perp}_{z}(t,F^*_{t})-U_{z}(t,F^*_{t})\eta_{t}$, we  conclude that $U_{z}(t,F^*_{t})$ is solution of the SDE \eqref{eqU_{z}}.
 \end{proof}


\section{Application to dynamic CRRA utilities}\label{sec:exCRRA}
We provide the resolution in the  important example of  dynamic Constant Relative Risk Aversion (CRRA) utilities, also called power dynamic utilities, introduced in \ref{DefDynUtility}. We now assume that the aggregated utility $V$  of pensioners is given by  an aggregation of  individual power dynamic utilities $\bar v$, with  pensioners having the same relative risk aversion  $\theta$. \\
The social planner needs to infer a dynamic utility $U$ of the fund such that her preference criterion $(U,V)$ is consistent. 
It is then  natural to search the utility $U$ in the same class\footnote{If  $U$ is power then necessarily $V$ is power  with the same risk aversion coefficient, to ensure  the consistency  of the criterion $(U,V)$ (see \cite{MrNek07}). }.
Therefore the goal is to find a consistent CRRA utility $U$ (if it exists)  satisfying assumptions of Theorem \ref{ThOpt}: 
\begin{equation}
\label{PowUtilitiesShiftKbis}
U(t,z) = Z_t^u \frac{(z-X_t)^{1-\theta}}{1-\theta},
\end{equation}
where $Z^u$ is assumed to be a positive process with dynamics
\begin{equation}
d Z_t^u = Z_t^u( b_t dt + \delta_t \cdot d W_t ) , \quad \delta_t \in \R^n.
\end{equation} 
As discussed in Remark \ref{Rk:sureplication},  to take into account the  sustainability constraint, $X$ should be  taken as a buffer fund with pension rate $P^{min}$ and contribution  rate $C$,  and  that super-replicates (pathwise)  the sustainability bound $\mathfrak{K}$.  The dynamics of the shift $X$ is, according to Theorem \ref{SusCond}, 
\begin{equation}\label{BFdynamicex1}
dX_t = (X_t r_t + C_t -   P^{min}_t) dt + \delta^X_t \cdot (dW_t+\eta_{t}dt) ,  \quad \delta^X \in \sigR.
\end{equation}

\noindent {  The dynamics  of the  utility process  $U$  is deduced easily from that of $Z^{u}$ and $X$.
\begin{Lemma}
\label{LemmaDynPopUShiftK}
Denoting $d X_t=  \mu_t^X dt + \delta_t^X . dW_t$, the dynamics of the shift $X$, 
the dynamics of the CRRA shifted utility  $U(t,z) = Z_t^{u} \frac{(z-X_t)^{1-\theta}}{1-\theta}$   is 
$dU(t,z) = \beta(t,z) dt + \gamma(t,z) .dW_t,$
with local characteristics 
\begin{numcases}{}
 \beta(t,z) =  - U_z(t,z)(\mu_t^X + \delta_t \cdot \delta_t^X)+ \frac{1}{2}U_{zz}(t,z)\vert\vert \delta^X_t\vert\vert^2  + U(t,z)b_t  \nonumber \\
\gamma(t,z) = -U_z(t,z)\delta^X_t + U(t,z)\delta_t  \nonumber 
\end{numcases}{}
\noindent If the shift $X$ is  a buffer fund (with pension rate  $P^{min}$ and contribution  rate  $C$) then \\  $\mu_t^X = X_t r_t + C_t -   P^{min}_t + \delta^X_t. \eta_{t}$ and  $\delta^X  \in \sigR.$
\end{Lemma}
\begin{proof}
The proof is a straightforward application of the It\^o Lemma. 
\end{proof}
\noindent We also provide some useful relations  for power  utilities and  their derivatives   of first  and second order and their    conjugate:
\begin{equation}\label{dualCRRA} 
	\left\lbrace
		\begin{aligned}
			(z - X_t) U_z(t,z)&=   Z^{u}_t  (z - X_t)^{1-\theta}= (1-\theta) U(t,z)  \\
			(z - X_t)^2U_{zz}(t,z)&= -\theta  (z - X_t)  U_z(t,z)  = -\theta (1-\theta) U(t,z)
		\end{aligned}
	\right.
\end{equation}

\noindent The first term $\frac{\theta}{1-\theta} (Z^{u}_t )^{(1/\theta)}   y^\frac{\theta-1}{\theta}$ of the dual  $\tU$  is standard, the second term $X_t y$ (linear in $y$) corresponds to  the  support function of the  convex domain $[ X_t, +\infty[ $.

\subsection{Age-independent pensions  } \label{calculsex1}
We start with Example \ref{Ex1}. Each pensioner receives the same pension  amount
$p_t=\rho_t  p_{t}^{min}$ (with $\rho \geq 1$),  and has the same dynamic power  utility $\bar v$,   
\begin{equation*}
{\bar v}(t,p) = Z_t \frac{(p-p^{min}_t)^{1-\theta}}{1-\theta}, 
\end{equation*}
where $Z$ is a positive adapted process.  The  social planner  aggregates preferences of  pensioners by weighting their utility  by  their  past contributions  $\alpha_c c_{t}(a)$, thus
\begin{equation}\label{PowUtilitiesShiftK}
V(t,\rho) =    \omega^r_t \bar v(s,\rho p^{min}_t) = \omega^r_t Z_t \; (p^{min}_t)^{1-\theta} \; \frac{(\rho-1)^{1-\theta}}{1-\theta}
\end{equation}
where  $\omega^r_t = \alpha_c   \int_{a_r}^{\infty}  c_{t}(a) n(t,a) da$ is the  total weight of all pensioners living at time $t$.

\subsubsection{Optimal policy}
\begin{Proposition} 
\label{OptimalStrategyCRRA}
Assume that  the shifted power utility   $U$ defined by \eqref{PowUtilitiesShiftKbis} is a consistent utility verifying assumptions of Theorem \ref{ThOpt}. Then  the optimal strategy is given by 
 \begin{numcases}{}
 \label{EqPiPow}
 \pi_t^* =\delta_t^{X} +  \frac{1}{\theta}(F_t^*- X_t) (\delta_t^{\mathcal{R}}+\eta_t) \\
p^{*}_t = p^{min}_t+  (F_t^*  - X_t) \left(\frac{Z_t \omega^r_t}{Z_t^u N^r_t}\right)^{\frac{1}{\theta}} 
\end{numcases}
and the optimal fund $F^*$ is solution of the dynamics
\begin{equation}\label{DynPuiF^*}
dF_{t}^*=d X_t + (F_t^* - X_t)  \left( \left(r_{t} {-}  (\frac{Z_t \omega^r_t}{Z_t^u N^r_t})^{\frac{1}{\theta}} \right)dt
+\frac{1}{\theta} (\delta_t^{\mathcal{R}}+\eta_t)  \cdot \big(dW_{t}+\eta_{t}dt\big)\right).
\end{equation}
\end{Proposition}
%
\noindent Note that due to  the form of the  pension utility $\bar v$, $\rho^f_t$  defined by \eqref{pf} is always greater than the bound 1 and thus  $\rho_t^*=\rho^f_t$  a.s.   \\
The optimal strategy has a particular additive form.
The optimal pension is the minimal pension $p_{t}^{min}$ plus an additional term that is proportional to the "cushion" $(F_t^*- X_t)$, where we recall that $X$ is a buffer fund receiving the contribution rate $C$ and paying the minimal pension amount $P^{min}$, which (sur)replicates the sustainability bound $\mathfrak{K}$. The proportionality factor $\left(\frac{Z_t \omega^r_t}{Z_t^u N^r_t} \right)^{\frac{1}{\theta}} $ is decreasing in the risk aversion parameter $\theta$.
The ratio  $\frac{ \omega^r_t}{ N^r_t}= \frac{  \int_{a_r}^\infty  \alpha_c c_{t}(a)  n(t,a)da }{ \int_{a_r}^\infty  n(t,a)da }$ represents  the  average past contributions of one pensioner at time $t$. With this choice of weight $w^r$, the more pensioners contributed during their working period, the higher their pension will be. This can be interpreted as some kind of actuarial fairness. \\
The ratio $\frac{Z_t }{Z_t^u }$ represents the relative importance of the pensioners utility with respect to the buffer fund utility, which represents future generations. This can be interpreted as an ``intergenerational risk sharing" coefficient. If $\frac{Z_t }{Z_t^u }>>1$, current pensioners will receive higher pensions, while   the fund return rate will decrease, potentially putting the  pension system's sustainability at risk  (see Theorem \ref{ThResult}).  On the other hand, if $\frac{Z_t }{Z_t^u }<<1$, then future pensioners will benefit from a higher fund value. \\
The optimal portfolio  is the  portfolio $\delta^X$   of buffer fund shift $X$  plus an additional term  proportional to the ``cushion" $(F_t^*- X_t)$. Again the  proportionality factor 
 $\frac{(\delta_t^{\mathcal{R}}+\eta_t)}{\theta}$ is decreasing in the risk aversion parameter $\theta$. Note that the optimal strategy also depends on the volatility of buffer fund utility. 
\begin{proof} By a simple derivation with respect to $z$ and $\rho$, we have 
$U_z(t,z) = Z_t^u (z-X_t)^{-\theta}$
 and $ V_\rho(t,\rho) = Z_t\omega_{s}^rp_{t}^{min}(p_{t}^{min}(\rho-1))^{-\theta} $
so that  ${p_{t}^{min} V_\rho^{-1}(t,y) = p_{t}^{min}+  \left(\frac{y}{Z_t\omega_{s}^rp_{t}^{min}}\right)^{-\frac{1}{\theta}}}.$
Thus, applying Theorem \ref{ThOpt} 
\begin{align*}
 p^*_t & = p_{t}^{min} \, V_\rho^{-1}(t, P_{t}^{min}U_z(t,F_t^*))\\
 & =p_{t}^{min}  \,  V_\rho^{-1}(t, P_{t}^{min}Z_t^u(F_t^*-X_t)^{-\theta} )\\
 & =p_{t}^{min}+ \left(\frac{Z_t^u(F_t^*-X_t)^{-\theta} P_{t}^{min}}{Z_t \omega^r_tp_{t}^{min}}\right)^{-\frac{1}{\theta}}\\
 & =p_{t}^{min}+ \left(\frac{Z_t \omega^r_t}{Z_t^u N^r_t} \right)^{\frac{1}{\theta}} (F_t^* -X_t). 
\end{align*}
where we used  $P_{t}^{min}=p_{t}^{min}N_t^r$ for the last equality. In addition,  the optimal investment strategy is given by 
\begin{align*}
\pi^{*}_t & =-\frac{\gamma^{\mathcal R}_{x}(t,F_t^*)+U_{z}(t,F_t^*) \eta_t}{U_{zz}(t,F_t^*)} \\
 & = -\frac{-U_{zz}(t,F_t^*)\delta_t^{X}+ U_z(t,F_t^*)(\delta_t^{\mathcal{R}}+\eta_t) }{U_{zz}(t,F_t^*)}\\
 & =  \delta_t^{X} + \frac{1}{\theta}(F_t^*- X_t) (\delta_t^{\mathcal{R}}+\eta_t) 
\end{align*}
where we used  Lemma \ref{LemmaDynPopUShiftK}. Plugging this strategy into the buffer fund dynamics  \eqref{EqFond}    yields \begin{eqnarray*}
dF_{t}^*&=&\Big(F_{t}^*r_{t}+C_{t}-N^r_t p^{min}_t \left( 1+ \left(\frac{Z_t \omega^r_t}{Z_t^u N^r_t}\right)^{\frac{1}{\theta}} (F_t^* - X_t)\right)\Big)dt\nonumber\\
&+&\big(\frac{1}{\theta}(F_t^*- X_t) (\delta_t^{\mathcal{R}}+\eta_t) +  \delta_t^{X}\big).\big(dW_{t}+\eta_{t}dt\big).
\end{eqnarray*} which implies \eqref{DynPuiF^*}  by  \eqref{BFdynamicex1}.
\end{proof}
 }

\subsubsection{Consistency}
 The next step consists in verifying  the assumptions of Theorem \ref{ThOpt} and in particular 
the consistency condition, that  translates into condition on $Z^u$. 
\begin{Proposition}
\label{PropEqZu}
Let $V$ be the aggregated  pensioners utility  \eqref{PowUtilitiesShiftK} and $U(t,z) = Z_t^u \frac{(z-X_t)^{1-\theta}}{1-\theta}$.
If  $(U,V)$ is consistent then  $Z^u$ must be  a well-defined solution  to the SDE:
{\fontsize{10.45}{6} \begin{equation}\label{Zu}
 -dZ^u_t = Z^u_t \left( \left( (1-\theta)r_t +   \frac{(1-\theta)}{2\theta}   \vert \vert { \delta_t^{\mathcal{R}}+ \eta_t}\vert \vert^{2}+    \theta (Z_t \omega_t^r)^{\frac{1}{\theta}  }  ({N^r_t })^{\frac{\theta-1}{\theta}} ({ Z_t^u})^{\frac{-1}{\theta}}  \right)dt -  \delta_t. dW_t \right).
 \end{equation}}
\end{Proposition}
\begin{proof}
One the one hand, by application of Lemma \ref{LemmaDynPopUShiftK}, the buffer fund dynamic utility  $U$ has the following dynamics
{\fontsize{10.45}{6}
\begin{equation}\label{beta+crra} 
	\left\lbrace
		\begin{aligned}
			&dU(t,z) = \beta(t,z) dt + \gamma(t,z) dW_t,   \\
			&\beta(t,z) =  - U_z(t,z)\big(X_tr_t + C_t -   P^{min}_t+ \delta^X_t \cdot ( \eta_{t} + \delta_t) \big)+ \frac{1}{2}U_{zz}(t,z)\vert\vert \delta^X_t\vert\vert^2  + U(t,z)b_t\\
			&\gamma(t,z) = -U_z(t,z)\delta^X_t + U(t,z)\delta_t.
		\end{aligned}
	\right.
\end{equation}
}
On the other hand,  if $(U,V)$ is consistent  the  consistency constraint \eqref{HJB-Beta} should be satisfied:
\begin{equation*}
 \beta(t,z) =  -U_{z}(t,z) \big(zr_t+C_t- P^{min}_t\rho_t^*(z)\big)
+\demi U_{zz}(t,z)\| \pi_t^*(z)\|^{2}-V(t, \rho_t^*(z)). 
\end{equation*}
By the expression of optimal controls given in Proposition \ref{OptimalStrategyCRRA},  the elementary identities \eqref {dualCRRA} and using   
\eqref{PowUtilitiesShiftK} and   \eqref{EqPiPow}
we can rewrite the consistency constraint \eqref{HJB-Beta} as 
 \begin{align*}
 \beta(t,z) & = - U_z(t,z)\big(zr_t + C_t -   P^{min}_t + \delta^X_t \cdot ( \eta_{t} + \delta_t) \big)+ \frac{1}{2}U_{zz}(t,z)\vert\vert \delta^X_t\vert\vert^2  \\
 & + U(t,z)\left((1-\theta) N_t^r (\frac{Z_t\omega_t^r}{Z^u_t N_t^r})^{\frac{1}{\theta}} - \frac{(1-\theta)}{2 \theta} \|\delta^{\mathcal{R}} + \eta_t \|^2 - \omega_t^r \frac{Z_t}{Z_t^u}(\frac{Z_t\omega_t^r}{Z^u_t N_t^r})^{\frac{1}{\theta}-1}\right).
\end{align*}
Identifying this with \eqref{beta+crra}, we obtain 
\begin{eqnarray*}
- U_z(t,z) X_tr_t +  U(t,z)b_t =  - U_z(t,z) zr_t &\\
&\hspace*{-6cm}+U(t,z)\left((1-\theta) N_t^r (\frac{Z_t\omega_t^r}{Z^u_t N_t^r})^{\frac{1}{\theta}} 
 - \frac{(1-\theta)}{2\theta} \|\delta^{\mathcal{R}} + \eta_t \|^2 -\omega_t^r \frac{Z_t}{Z_t^u}(\frac{Z_t\omega_t^r}{Z^u_t N_t^r})^{\frac{1}{\theta}-1}\right)
\end{eqnarray*}
which implies that the drift $b_t$ of $Z^u$ must satisfy
\begin{equation*}
b_t = - (1-\theta)r_t  - \frac{(1-\theta)}{2\theta} \|\delta^{\mathcal{R}} + \eta_t \|^2 - \theta(N_t^r)^{1-\frac{1}{\theta}} \big(\frac{\omega_t^rZ_t}{Z^u_t}\big)^{\frac{1}{\theta}} .
\end{equation*}
\end{proof}

\noindent    SDE  \eqref{Zu} is only defined on $Z^u >0$, due to the term $(Z_t^u)^{-\frac{1}{\theta}}$ in the drift.  %
\begin{Proposition}
\label{Proptau}
Let $\xi$ be the process defined by 
\begin{equation*}
\xi_t = \exp\left(\int_0^t \big((1-\theta)r_t + \frac{(1-\theta)}{2\theta}   \vert \vert { \delta_t^{\mathcal{R}}+ \eta_t}\vert \vert^{2} +\frac{\|\delta_s\|^2}{2}\big)ds - \int_0^t \delta_s\cdot  dW_s\right)
\end{equation*}
There exists a unique solution of the SDE \eqref{Zu}, defined by 
\begin{equation}
\label{exprZu}
Z_t^u = \xi_t^{-1}\left(Z_0^{\frac{1}{\theta}} -  \int_0^t (N_s^r)^{1-\frac{1}{\theta}}(Z_s \omega_s^r\xi_s)^{\frac{1}{\theta}} ds \right)^{\theta}, \quad \forall t \in [0,\tau ^Z[, 
\end{equation}
where $\tau^Z$ the first hitting time of 0 satisfies
\begin{equation}\label{deftau}
\tau^Z = \inf \{ t \; ; \;  \int_0^t (N_s^r)^{1-\frac{1}{\theta}} (\frac{Z_s}{ Z_0}\omega_s^r \xi_s )^{\frac{1}{\theta}}ds  \geq1 \}.
\end{equation}
\end{Proposition}
\begin{proof}
Let $Z^u$ be a solution of \eqref{Zu} on $[0,\tau[$, with $\tau = \inf\{ t\geq 0 ; \; Z_t^u =0\}$, and denote
\begin{align*}
\tilde{r}_t := (1-\theta)r_t + \frac{(1-\theta)}{2\theta}   \vert \vert { \delta_t^{\mathcal{R}}+ \eta_t}\vert \vert^{2},\quad \text{ and }  \tilde \beta_t := (N_t^r)^{1-\frac{1}{\theta}}(Z_t \omega^r_t)^{\frac{1}{\theta}}, 
\end{align*}
so that 
\begin{equation*}
dZ_t^u = -\big(Z_t^u\tilde{r}_t  +\theta  \tilde \beta_t (Z_t^u)^{1-\frac{1}{\theta}}\big)dt   + Z_t^u \delta_t \cdot dW_t .
\end{equation*}
This equation may be simplified by the following change of variables: 
\begin{equation*}
\bar{Z}_t := Z_t^u\exp\left(\int_0^t \big(\tilde r_s +\frac{\|\delta_s\|^2}{2}\big)ds - \int_0^t \delta_s\cdot  dW_s\right) =  Z_t^u \xi_t
\end{equation*}
Then,
\begin{equation*}
\frac{1}{\theta}(\bar{Z}_t)^{\frac{1}{\theta}-1}d\bar{Z}_t = - \beta_t \xi_t^{\frac{1}{\theta}}dt, 
\end{equation*}
i.e 
\begin{equation*}
(\bar{Z}_t)^{\frac{1}{\theta}} - (\bar{Z_0})^{\frac{1}{\theta}} = - \int_0^t \beta_s\xi_s^{\frac{1}{\theta}} ds, 
\end{equation*}
and thus
\begin{equation*}
Z_t^u = \xi_t^{-1}\left(Z_0^{\frac{1}{\theta}} -  \int_0^t (N_s^r)^{1-\frac{1}{\theta}}(Z_s\omega_s^r \xi_s)^{\frac{1}{\theta}} ds \right)^{\theta}, 
\end{equation*}
and $\tau^Z$ verifies \eqref{deftau}.
\end{proof}
\noindent The necessary condition for $Z^u$ given in Proposition \ref{PropEqZu} is actually a sufficient condition. 

\subsubsection{Existence of the optimal investment/pension policy}
\begin{Theorem}
\label{ThResult}
Let $Z^u$ be defined as in Proposition \ref{Proptau}. Then the shifted power dynamic utilities $(U,V)$ given by \eqref{PowUtilitiesShiftKbis}-\eqref{PowUtilitiesShiftK} are consistent on $[0,\tau^Z[$, and verify the assumption of Theorem \ref{ThOpt} on this interval. Therefore the optimal strategy is given on $[0,\tau^Z[$ by 
 \begin{numcases}{}
\nonumber \pi_t^* =  \delta_t^{X} +  \frac{1}{\theta}(F_t^*- X_t) (\delta_t^{\mathcal{R}}+\eta_t) \\
\nonumber p^{*}_t = p^{min}_t +   (F_t^*  - X_t) \left(\frac{Z_t \omega^r_t}{Z_t^u N^r_t}\right)^{\frac{1}{\theta}}.
\end{numcases}
Furthermore, $\tau^Z  = \inf \{t \geq 0 ; \; F^*_t = X_t \}$. 
\end{Theorem}
\noindent Observe that $p^*$  depends on the population dynamics only through the ratio $\omega_t^r/N^r_t$. Since all pensioners receive the same pension amount $p^*_t$, there is also  an intergenerational risk-sharing among pensioners. \\
For $t\geq \tau^Z$, $Z_t^u=0$ and thus the buffer fund utility $U(t,\cdot)\equiv 0$. At this stopping time, the buffer fund hits the boundary $X$.  In order to stay sustainable, the pension system parameters have to be updated, for instance by decreasing the sustainability bound  $\mathfrak{K}_t$,  the minimum pension amount $p^{min}_t(a)$ or the importance given to the pensioners' preferences. Observe that thanks to  \eqref{deftau}, the distribution $\tau^Z$ can be approximated numerically. In particular,  $\tau^Z$ decreases with the number of pensioners and the weight of  pensioner's preferences attributed by the social planner.  However,  $\tau^Z$  does not depend on the initial buffer fund amount $F_0$, due to the choice of dynamic power utilities.
\begin{proof}[Proof of Theorem \ref{ThResult}]
The consistency SDE is verified by Proposition \ref{PropEqZu} and \ref{Proptau}. It is straightforward to check Assumption \ref{HypBarU}, since $\bar{U}(t,z) = U(t,z-X_t)  = \frac{z^{1-\theta}}{1-\theta}$. \\
Furthermore,  $\tilde{V} (t,z) =z+\frac{\theta}{1-\theta} (\omega_t^r Z_t (p^{min}_t)^{1-\theta})^{\frac{1}{\theta}}z^{1-\frac{1}{\theta}}= z+\frac{\theta}{1-\theta} (\omega_t^r Z_t)^{\frac{1}{\theta}}\left(\frac{z}{p^{min}_t}\right)^{1-\frac{1}{\theta}}$, and it is straightforward to check that Assumption \ref{hypVbis} is verified.  The optimal strategy is then obtained from 
Proposition \ref{OptimalStrategyCRRA}.  
By Proposition \ref{OptimalStrategyCRRA}, and using the short notation  $\mathcal{E}(\int_0^t g_s \cdot d W_s)$ for  the Doléans-Dade exponential martingale\footnote{  $\mathcal{E}(\int_0^t g_s \cdot d W_s):=  \exp  \left(\int_0^t g_s \cdot d W_s- \demi \int_0^t   \vert \vert g_s  \vert \vert^2 ds \right)  $ .     }
$$(F_{t}^* - X_{t})=\big(F_{0}^* - X_{0}\big)
\exp \Big(\int_0^t (r_{s}-(N^r_s)^{1-\frac{1}{\theta}} (\frac{Z_s \omega_s^r}{Z_s^u})^{\frac{1}{\theta}  } +  \frac{1}{\theta}(\delta_s^{\mathcal{R}}+\eta_s)\cdot \eta_s \big)ds\Big)
\mathcal E \Big(\int_0^t   \frac{1}{\theta}(\delta_s^{\mathcal{R}}+\eta_s\big) \cdot dW_{s}\Big),$$
and by Proposition \ref{PropEqZu}, 
$$Z^u_t  =  Z^u_0 \exp  \Big(- \int_0^t  \left( ( (1-\theta)r_s +   \frac{(1-\theta)}{2\theta}   \vert \vert { \delta_s^{\mathcal{R}}+ \eta_s}\vert \vert^{2}+    \theta (N^r_s)^{1-\frac{1}{\theta}} (\frac{Z_s \omega_s^r}{Z_s^u})^{\frac{1}{\theta}  }  ds  \Big)  \mathcal E \Big(\int_0^t  \delta_s \cdot dW_s \right).$$
Then, 
$$(F_{t}^* - X_{t})(Z^u_t )^{-\frac{1}{\theta}}=\big(F_{0}^* - X_{0}\big)(Z^u_0)^{-\frac{1}{\theta}  }
\exp \Big( \frac{1}{\theta}\int_0^t (r_{s} +  \frac{1 }{2}   \vert \vert { \delta_s^{\perp}- \eta_s}\vert \vert^{2}) ds
- \int_0^t   \frac{1}{\theta}(\delta_s^{\perp}-\eta_s) \cdot  dW_{s}\Big).$$
The previous equation can be rewritten as 
$$(F_{t}^* - X_{t})(Z^u_t )^{   \frac{-1}{\theta}}=\big(F_{0}^* - X_{0}\big)(Z^u_0)^{   \frac{-1}{\theta}}  (Y_t)^{   \frac{-1}{\theta}}, $$
where we recognize the state price density process $Y$ 
 $$Y_t= \exp \left(-  \int_0^t r_{s} ds \right) \mathcal{E} \Big( \int_0^t   (\delta_s^{\perp}-\eta_s) \cdot  dW_{s}\Big).$$
In particular $\tau^Z  = \inf \{t \geq 0 ; \; F^*_t = X_t \}$.
\end{proof}

\noindent Section \ref{calculsex2} explains how to extend the results obtained  for   Example \ref{Ex1} to  Example \ref{Ex2}. 

\subsection{Age-dependent  pensions}\label{calculsex2}

  In Example  \ref{Ex2}   of age-dependent  pension and  utility $\bar v$,   the aggregate utility of pension $V$  is a complex aggregation between cohorts given by \eqref{EqVexample2}
$$V(t,\rho) = \int^\infty_{a_r} \bar v(t, a, \rho \, p_{ret}(a_r + t-a)e^{\int_{a_r+t-a}^t\lambda_u du} ) \, n(t,a) da.$$
We assume that  pensioner utility depends on the age $a$ through the shift $p^{min}(t,a) $ in the shifted power utility
$$\bar v(t,a,p) = Z_t \frac{(p-p^{min}_{t}(a) )^{1-\theta}}{1-\theta} \quad \mbox{  and }  \quad v(t,a, \rho) =  (p^{min}_{t}(a))^{1-\theta}  Z_t  \frac{(\rho - 1)^{1-\theta}}{1-\theta}.$$
Therefore in  the  case  of power utility, the aggregate  utility of pension has  the following multiplicative form
$$V(t,\rho)= \tilde{\omega}^r_t Z_t \frac{(\rho - 1)^{1-\theta}}{1-\theta}, \quad \quad 
\mbox{ with } \quad  \tilde{\omega}^r_t= \int_{a_r}^\infty  (p^{min}_{t}(a))^{1-\theta}   n(t,a) da.$$
Observe that the problem is formulated   similarly than  in  Example  \ref{Ex1}, with a different  weight $\tilde{\omega}^r_t$ 
and  in this case $P_{t}^{min}=\int_{a_r}^\infty  p^{min}(t,a)   n(t,a) da$. Thus similar computations as in Section  \ref{calculsex1} yield
\begin{align*}
 \rho^*_t & = V_\rho^{-1}(t, P_{t}^{min}U_z(t,F_t^*))\\
 & =  V_\rho^{-1}(t, P_{t}^{min}Z_t^u(F_t^*-X_t)^{-\theta} )\\
 & =1+ \left(\frac{Z_t^u(F_t^*-X_t)^{-\theta} P_{t}^{min}}{Z_t \tilde \omega^r_t}\right)^{-\frac{1}{\theta}}\\
 & =1+  (F_t^* -X_t) \left(\frac{Z_t   }{ Z_t^u}\frac{ \int_{a_r}^\infty  (p^{min}_{t}(y))^{1-\theta}   n(t,y) dy} {  \int_{a_r}^\infty p^{min}_{t}(y)   n(t,y) dy }\right)^{\frac{1}{\theta}} . 
\end{align*}
Therefore the optimal strategy in this Example \ref{Ex2}  with CRRA utility is (on  $[0,\tau^Z[$) 
\begin{numcases}{}
 \label{EqPiPowex2}
\nonumber \pi_t^* =  \delta_t^{X} +  \frac{1}{\theta}(F_t^*- X_t) (\delta_t^{\mathcal{R}}+\eta_t) \\
\nonumber p^{*}_{t}(a) = p^{min}_{t}(a)\Big( 1 +    (F_t^*  - X_t)  \left(\frac{Z_t   }{ Z_t^u}\frac{ \int_{a_r}^\infty  (p^{min}_{t}(y))^{1-\theta}   n(t,y) dy} {  \int_{a_r}^\infty p^{min}_{t}(y)   n(t,y) dy } \right)^{\frac{1}{\theta}} \Big).
\end{numcases}
Observe that in Examples \ref{Ex1} and \ref{Ex2},  the portfolio are the same.
What differs are the pensions. Nevertheless if
$p^{min}_{t}(a)=p^{min}_t$ does not depend on the age, then  the optimal pension $p^{*}_{t}(a)$ simplifies in $p^{*}_{t}(a)=p^{*}_t=  p^{min}_t +   (F_t^*  - X_t) \left(\frac{Z_t   }{ Z_t^u}\right)^{\frac{1}{\theta}}$. It is coherent with Example \ref{Ex1} since it corresponds to the particular case of Example  \ref{Ex1} in which each pensioner has weight 1, that is $\omega_t^r=N_t^r$.\\
 The increase of the pension amount with age depends on the indexation rate $\lambda_t$ in $p^{min}_t(a)$ and on the optimal adjustment $\rho^*_t = 1 +    (F_t^*  - X_t)  \left(\frac{Z_t   }{ Z_t^u}\frac{ \int_{a_r}^\infty  (p^{min}_{t}(y))^{1-\theta}   n(t,y) dy} {  \int_{a_r}^\infty p^{min}_{t}(y)   n(t,y) dy } \right)^{\frac{1}{\theta}} $. As in  Example~ \ref{Ex1}, $\rho^*$ still depends on the  importance attributed to the  pensioners' preferences  with respect to  $Z^u$. The pension can also be written as follows, to make appear the relative weight of cohort $a$, captured through $p^{min}_{t}(a)$, with respect to the other cohorts:
 $$ p^{*}_{t}(a) = p^{min}_{t}(a) +    (F_t^*  - X_t)  \left(\frac{Z_t   }{ Z_t^u}\frac{  \int_{a_r}^\infty (\frac{p^{min}_{t}(y)}{p^{min}_{t}(a)} )^{1-\theta}    n(t,y) dy }{ \int_{a_r}^\infty  (\frac{p^{min}_{t}(y)}{p^{min}_{t}(a)} )   n(t,y) dy} \right)^{\frac{1}{\theta}}.
$$

\bigskip

\paragraph{Conclusion} This paper  designs  a social planner's dynamic decisions criterion   under sustainability, adequacy and fairness constraints.
The  optimal   investment/pension policy is derived when the social planner can invest in/borrow from a buffer fund,  with the aim to  provide a better demographic and financial risk-sharing across generations.  This flexible modeling can be easily extended  to heterogenous cohorts or open populations, thus considering also intra generational risk sharing.\\
The explicit computations for the optimal policies allows   these  theoretical results to be applied to an empirical setting with real data. For instance, an interesting application will consist in computing  and analyzing  actuarial fairness criteria for each generation,  and to test the impact of demographic shocks (such as the "baby boom") on optimal contribution/benefit plans.


\bibliographystyle{alpha}
\bibliography{RandomRef}

 \end{document}